\documentclass[11pt,letterpaper]{llncs}
\usepackage{soul}
\usepackage[utf8]{inputenc}
\usepackage[table]{xcolor}
 \usepackage{cancel} 
\usepackage{float}
\usepackage[margin=1in]{geometry}
\usepackage{graphicx}
\usepackage{xcolor}
\usepackage{amsmath}
\usepackage{amssymb}
\usepackage[sort,nocompress]{cite}
\usepackage[ruled,vlined,lined,boxed,commentsnumbered]{algorithm2e}
\usepackage{todonotes}
\usepackage{verbatim}
\usepackage{subcaption}
\usepackage[normalem]{ulem}
\usepackage{array}
\usepackage{algorithmic}
\usepackage{mathrsfs}
 \usepackage{multirow}

\usepackage{todonotes}

\usepackage{pgf}
\usepackage{tikz}

\newcommand{\pf}[1]{\text{Pref} (\mathcal #1)}
\newcommand{\pff}[1]{\text{\emph{Pref}}(\mathcal #1)}

\newcommand{\nats}{\mathbb N} 

\newcommand{\mt}[1]{\mathcal #1}

\newcommand{\notaN}[1]{\textcolor{brown}{#1 NP}}

\usetikzlibrary{arrows,automata,positioning}
\title{On (co-lex) Ordering Automata}

\author{Nicola Cotumaccio\inst{1} \and Giovanna D'Agostino\inst{2} \and Alberto Policriti\inst{2}  \and Nicola Prezza\inst{3}} 

\institute{Gran Sasso Science Institute, L'Aquila,  Italy. Email: \email{nicola.cotumaccio@gssi.it} \and University of Udine, Italy.  Email: \email{giovanna.dagostino@uniud.it, alberto.policriti@uniud.it} \and University Ca' Foscari, Venice,  Italy. Email: \email{nicola.prezza@unive.it}}

\date{\today}

\begin{document}

\maketitle
\thispagestyle{empty}

\begin{abstract}
   
The states of a deterministic finite automaton $\mathcal A$ can be identified with collections of words in $\pf{L(\mathcal A})$—the set of prefixes of words belonging to the regular language accepted by $\mathcal A$. But words can be ordered and among the many possible orders a very natural one is the co-lexicographic one. Such naturalness stems from the fact that it suggests a transfer of the order from words to the automaton’s states. This suggestion is, in fact, concrete and in a number of papers  automata admitting a total ordering of states coherent with the ordering of the set of words reaching them  have been proposed, studied, and implemented. Such class of ordered automata—the Wheeler automata—turned out to be efficiently stored/searched using an index and the set of words reaching each of their states have been proved to be intervals of the co-lexicographic ordered $\pf{L(\mathcal A})$.

Unfortunately  not all automata can be totally ordered as previously outlined. That is, in general not all sets of prefixes of $\mathcal L(\mathcal A)$ reaching a given state form an interval in $\pf{L(\mathcal A})$. However, automata can always be \emph{partially} ordered and an intrinsic measure of their complexity can be defined and effectively determined, as the minimum width of  one of their admissible partial orders.
As    shown in previous works, this new concept of width of an automaton  has useful consequences in the fields of graph compression, indexing data structures, and automata theory.

In this paper we prove that  a canonical, minimum-width, partially-ordered, automaton accepting a language $\mathcal L$—dubbed the Hasse automaton $\mathcal H$ of $\mathcal L$—can be exhibited. $\mathcal H$ provides, in a precise sense, the best possible way to (partially) order the states of any automaton accepting $\mathcal L$, as long as we want to maintain an operational link with the (co-lexicographic) order of $\pf{L(\mathcal A})$. 
The difficulty in defining the Hasse automaton is a consequence of the existence of (co-lex) monotone sequences  of words  reaching infinitely many times more than one state of any deterministic automaton accepting L. Properly constructing the Hasse's automaton requires a distribution of these monotone sequences upon $\mathcal H$’s states on different linear components of the underlying partial order.
The number of such linear components (i.e. the width of the underlying partial order) intrinsically constrains the index-ability and, hence, the complexity of $\mathcal L$. 
We prove that such a complexity, that is,  the width of the language, can be effectively computed from the minimum  automaton recognizing the language. 
Finally, we  explore the relationship between two (often conflicting) objectives: minimizing the width and minimizing the number of states of an automaton.

\end{abstract}


\newpage

\clearpage
\pagenumbering{arabic} 
 
 \section{Introduction}

Equipping the domain of a structure with some kind of order is often a fruitful move performed in both Computer Science and Mathematics. An order provides direct access to data or domain elements and sometimes allows to tackle problems too difficult to cope with without some direct manipulation tool. For example, in descriptive complexity it is not known how to logically capture the class $P$ in general, while this can be done using a logic defined on a signature including an order relation (see \cite{Libkin}). 
The price to be paid  when imposing an order may be the restriction of the class of admissible structures. If we do not wish to pay such a price, a \emph{partial} order can be a reasonable alternative. Then the farther the partial order is from a total order, the less we are able to simplify our study. The ``distance'' from the partial to a total order becomes    a measure of the extent to which we have to ``tame''  the class of structures under consideration.

In this paper we cast the above observations onto the class of finite automata.
Partial orders and automata have already met and attracted attention because of their relation with logical, combinatorial, and algebraic characterization of languages (see, among many others,  \cite{thierrin1974ordered_aut,  DBLP:journals/jcss/BrzozowskiF80, DBLP:conf/dlt/SchwentickTV01}). In the literature a partially-ordered NFA is an automaton where the transition relation induces a partial order on its states (see  \cite{DBLP:journals/lmcs/MasopustK21}, for example,  for recent complexity results on this class of automata). More recently, in an effort to find a simple and unitary way to present a number of different algorithmic techniques related to the Burrows-Wheeler transform \cite{burrows1994block}, in  \cite{GAGIE201767}  Gagie et al.    proposed a different approach. The simple strategy they put forward was based on the idea of enforcing and using  order on a given automaton, based on an \emph{a priori} fixed order of its alphabet.
The resulting kind of automata, called Wheeler automata, admit an efficient index data structure for searching  subpaths labeled with a given query pattern  and    enable a representation of the graph in a space proportional to that of the edges' labels (as well as enabling more advanced compression mechanisms, see \cite{DBLP:conf/dcc/AlankoGNB19,DBLP:conf/soda/Prezza21}). 
This is in contrast with the fact that general graphs require a logarithmic (in the graph's size) number of bits per edge to be represented, as well as with recent results showing that in general, the subpath search problem can not be solved in subquadratic time, unless the strong exponential time hypothesis is false \cite{DBLP:journals/corr/BackursI15,DBLP:conf/sofsem/EquiMT21,DBLP:journals/corr/abs-1901-05264,DBLP:conf/sosa/GibneyHT21, DBLP:journals/ipl/PotechinS20}.
Wheeler languages---i.e. languages accepted by Wheeler automata---form an interesting   class of languages, where determinization and  membership  verification   become   easy tasks. However, as was to be expected, requiring the presence of a \emph{total} Wheeler order  over an automaton comes with a price: languages recognized by Wheeler automata constitute a small class, a subclass of star-free languages  \cite{alanko2020wheeler}.
 
Our results show that, as a matter of fact,  Wheeler (automata and) languages can be seen as a first level---in some sense the optimal level---of a hierachy of languages based on the minimum \emph{width} of a suitable \emph{partial} order over the states of  automata accepting the language. In other words, the minimum \emph{width} of a partial order definable on the collection of states is a measure that (properly, as we   prove in a companion paper in preparation \cite{DBLP:journals/corr/abs-2102-06798}), stratifies finite automata and regular languages. 

The  non-deterministic    width   of a language $\mathcal L$, ${\text{width}^N}(\mathcal L)$,   is  defined as the smallest $p$ such that there exists an NFA $\mathcal A$ with ${\text{width}}(\mathcal A)=p$  recognizing the language. The deterministic width of a language, ${\text{width}^D}(\mathcal L)$ is similarly defined using DFAs. 
The  non-deterministic/deterministic width allow to define two hierarchies of regular languages and in a     companion paper in preparation \cite{DBLP:journals/corr/abs-2102-06798}  we proved that   the  two hierarchies are  strict and do  not collapse. Except for level one---the Wheeler languages---the corresponding levels in the two hierarchies do not   coincide:  
there exist    languages $ \mathcal{L} $ such that $\text{width}^{D}(\mathcal L) > \text{width}^{N}(\mathcal L)$.
Moreover, we proved that   for a regular language $\mathcal L$ it holds:
  \[\text{width}^N(\mathcal L)\leq  \text{width}^D(\mathcal L) \leq 2^{\text{width}^N(\mathcal L)} - 1. \]

As shown in \cite{NicNic2021}, this new concept of width of an NFA/DFA has groundbreaking consequences in the fields of graph compression, indexing data structures, and automata theory: consider a finite automaton of width $p$. Then: (i) the automaton can be compressed using just $\Theta(\log p)$ bits per edge (assuming a constant-sized alphabet), (ii) pattern matching on the automaton's path (including  testing membership in the automaton's accepted language) can be solved in time proportional to $p^2$ per matched character, and (iii) the well-known  explosion in the number of states occurring when computing the powerset DFA equivalent to an input NFA is exponential in $p$, rather in the input's size \cite{NicNic2021}.

Having established that the width of a language is a meaningful complexity measure,  in this paper we address the problem of the effectiveness of this measure in the deterministic case:
is the deterministic width of a language $\mathcal L$, presented  by an automaton,  computable?
Notice that the width of a language  is not in general equal to the width of, say,  its  minimum automaton, since already at level one of the deterministic hierarchy there are Wheeler languages whose minimum automaton has deterministic width greater than one \cite{DBLP:conf/soda/AlankoDPP20}. 

\medskip

In this paper,  answering positively the above question we prove that \emph{any} language admits a canonical automaton realising its (best) width. 

In particular, we show that, although the deterministic width of a language, $\text{width}^D(\mathcal L)$, differs in general from the width of its (unique) minimum automaton,  $\text{width}^D(\mathcal L)$ can be computed and an automaton realising this width can be exhibited. The key observation for this result is a  combinatorial   property of automata that we called the {\em entanglement} number of a DFA. The entanglement of an automaton $\mathcal A$ measures the intrinsic incomparability of the automaton's states, which cannot be reduced even if we duplicate states.
This turns out to exactly correspond to the width of the partial order of sets of words reaching any given state---the partial order obtained as a lifting to sets of the co-lexicographic order of their elements/words. For this reason, the above mentioned canonical automaton is dubbed \emph{Hasse} automaton and its states can be distributed in $\text{width}^D(\mathcal L)$ linear components.   

Next, we prove that the entanglement of a language $\mathcal L$ (i.e. the best possible entanglement for $\mathcal L$) is readable/computable on the minimum DFA accepting $\mathcal L$ and can always be realized as the width of the Hasse automaton. The Hasse automaton is built on top of a  convex decomposition of the \emph{trace} of $\mathcal L$, which is the linear co-lexicographic order of the prefixes of words belonging to $\mathcal L$. 

Finally, we provide a DFA-free interpretation of the width by proving a full Myhill-Nerode theorem. Every linear component determines a partition of some sets of strings into convex sets, and algebraically we can capture this property by considering convex equivalence relations. We also exhibit a minimum automaton and, more generally, we explore the relationship between two (often conflicting) objectives: minimizing the width and minimizing the number of states.

\section{Notations on automata and orders}

Let $\Sigma$ be a finite alphabet and let  $ \Sigma^*$ be the set of  all (possibly empty) finite words on $\Sigma$. A finite  nondeterministic  automaton  (a NFA) accepting strings in $ \Sigma^{*} $ is a tuple  $\mathcal A=(Q, s, \delta,F)$ where $ Q $ is a finite set of states, $ s $  is a \emph{unique} initial state, $ \delta(\cdot, \cdot): Q \times \Sigma \rightarrow \mathcal Pow(Q) $ is the  transition function  (where $\mathcal Pow(Q)$ is the set of  all subsets of $Q$), and $ F\subseteq Q $ is the set of  final states. Sometimes we write $Q_{\mt A}, s_{\mt A}, \delta_{\mt A}, F_{\mt A}$ when the automaton $ \mathcal{A} $ is not clear from the context.

 As customary, we extend $ \delta $ to operate on strings as follows: for all $ q \in Q, a\in \Sigma, $ and $ \alpha \in \Sigma^{*} $:
\begin{align*}
 {\delta}(q,\epsilon)=\{q\}, &   \hspace{1cm}
 {\delta}(q,\alpha a) = \bigcup_{v\in{\delta}(q,\alpha)} \delta(v,a).
\end{align*}   
 We denote by   $\mathcal L(\mathcal A)=\{\alpha\in \Sigma^*: \delta(s, \alpha)\cap F\neq \emptyset\}$   the language accepted by the automaton $ \mathcal A $.
In this paper we will mostly consider {\em deterministic}  finite  automata (DFA),  where $ |\delta(q,a)|\leq 1 $,  for any $ q \in Q $ and $ a \in \Sigma $.
If the automaton is deterministic we  write  ${\delta}(q,\alpha)=q'$  for   the unique $q'$  such that ${\delta}(q,\alpha)=\{q'\}$  (if defined).

We assume from each state one can reach a final state (possibly the state itself)  and also that every state is reachable from the (unique) initial state. 
Hence, $ \text{Pref}(\mathcal L(\mathcal A))$, the collection of prefixes of words accepted by $ \mathcal A $,  consists of the set of  words that can be read on $ \mathcal{A} $ starting from the initial state. 
If $q\in Q$, $\alpha \in \pf L$, we denote by $I_q$ the set of words arriving at $q$: 
\[I_q=\{\alpha \in\pf L: q\in \delta(s, \alpha)\}.\]

In the paper we will   use the following   consequence of the Myhill-Nerode Theorem for regular languages: if $\mt A$ is the minimum DFA recognizing a language $\mt L$  and $\mt B$ is   any DFA  recognizing $\mt L$, then any $I_q$ with $q\in Q_{\mt A}$ is the union of a finite number of sets $I_u$, for $u\in Q_{\mt B}$.

A  {\em partial order} is a pair $ (Z, \le) $,   where     $ Z $ is a set and   $ \le $ is a binary relation on $ Z $ being reflexive, antisymmetric, and transitive. Any $  u, v \in Z $ are said to be $ \le $-{\em comparable} if either $ u \le v $ or $ v \le u $. We write $ u < v $ when $ u \le v $ and $ u \not = v $.
We write $ u ~\|~ v $ if $ u $ and $ v $ are not  $\leq$-comparable.  

A  partial order $ (Z, \le) $ is a {\em total order}     if     every pair of elements in $(Z, \le)$ are $\le$-comparable.   A subset $ Z' \subseteq Z $ is a $\le$-{\em chain} if $ (Z', \le) $ is a total order, and a family $ \{Z_i\}_{i = 1}^p $ is a $\le$-{\em chain partition} if $ \{Z_i\}_{i = 1}^p $ is a partition of $ Z $ and each $ Z_i $ is a $\le$-chain. 

The {\em width} of a partial order  $ (Z, \le) $ is the cardinality of a smallest $ \le$-chain partition.  A subset  $ U \subseteq Z $ is an $\le$-{\em antichain} if every pair of elements in $ U $ are not $ \le $-comparable. Dilworth's Theorem \cite{dilworth} states that the width of $ (Z, \le) $ is  equal to the cardinality of a largest $\le$-antichain.

A partial order   $\leq$ on $Z$   induces a   partial order over the  non-empty subsets of $Z$, still denoted by $\leq$ and defined   as follows:
$X\leq Y \Leftrightarrow ( X=Y) ~\lor~  \forall x \in X ~ \forall y\in Y ~( x<y)$.   If $ z \in Z $ and $ X \subseteq Z $, then $ z < X $ is a shortcut for $ \{z \} < X $ and $ X < z $ is a shortcut for $ X < \{z\} $.

A {\em monotone sequence} in a partial order $(Z, \le)$ is a sequence $(v_n)_{n\in \mathbb N}$ with $v_n\in V$ and either $v_i\le v_{i+1}$, for all $i$, or  $v_i\ge v_{i+1}$, for all $i$. 

Throughout the paper, we assume that  there is a fixed total order $ \preceq $ on $ \Sigma $ (in our examples, the alphabetical order), and we extend it \emph{co-lexicographically} to words in $\Sigma^*$, that is, $ \alpha \prec \beta $ if and only if the the reversed string $ \alpha^R $ is lexicographically smaller than the reversed string $ \beta^R $. 

If $\alpha\preceq  \alpha'\in \Sigma^*$, we define $[\alpha, \alpha']=\{\beta: \alpha\preceq \beta \preceq \alpha'\}$;  if the order among $\alpha, \alpha'$ is not known,  we set 
$[\alpha, \alpha']^{\pm}= 
[\alpha, \alpha']$, if $\alpha \preceq  \alpha'$, while  $[\alpha, \alpha']^{\pm}= 
[\alpha', \alpha]$,  if $\alpha' \preceq  \alpha$.

\section{Previous results}

In this section we contextualise  our work with respect to previously proved results in the field.

Wheeler automata/languages where introduced in \cite{GAGIE201767} as a means for transferring known and efficient string-manipulation indexes and techniques based on the Burrows-Wheeler transform \cite{burrows1994block} from strings  to languages. Given an ordered alphabet $ (\Sigma, \preceq) $, a  {\em Wheeler order} over an NFA  $\mathcal A=(Q, s, \delta, F)$ is   a \emph{total}  order  $ \leq $  over 
$Q$ in which  the initial state $s$  is the minimum and for every 
$ v_{1} \in \delta(u_{1},a_{1}) $ and $ v_{2} \in \delta(u_{2},a_{2}) $, with $ a_1, a_2 \in \Sigma $:
 \begin{enumerate}
	\item[(i)] $ a_{1}\prec a_{2}\rightarrow v_{1} < v_{2} $;
	\item[(ii)] $  (a_{1}=a_{2}  \wedge u_{1} < u_{2}) \rightarrow v_{1} \leq v_{2}$. 
\end{enumerate} 
As stated in the introduction, the class of Wheeler NFAs (i.e., NFAs endowed with a Wheeler order)  reduces membership-query time  from $O(\pi m)$ (where $\pi$ is the length of the query and $m$ is the size of the NFA)  to the optimal $O(\pi)$ after a suitable pre-processing (construction of an index) that can be done in time $O(m)$.    In \cite{DBLP:conf/esa/GibneyT19} the problem of deciding whether a given NFA admits a Wheeler order was proved to be NP-complete, while the same problem restricted to DFA or even to \emph{reduced} NFA (that is automata where different states are reached by different sets of words) becomes polynomial \cite{alanko2020wheeler}. A regular language is {\em Wheeler} if it is recognized by a Wheeler NFA. In \cite{DBLP:conf/soda/AlankoDPP20} it is proved that every Wheeler language is also recognized by some Wheeler \emph{deterministic} automaton. In \cite{alanko2020wheeler} it is also proved that it can be decided in polynomial time whether a regular language - given by means of a DFA recognizing the language - is Wheeler.  
Notice that the natural  generalization  obtained by considering all possible orders on the alphabet $ \Sigma $   leads to a larger class of languages, but, as proved in \cite{dagostino2021ordering},  the membership problem becomes  computationally more complicate ($NP$-complete instead of polynomial).

In \cite{NicNic2021} the totality  requirement  on  the order $\leq$ of states is dropped, thereby obtaining  the notion of {\em co-lexicographic order} over an NFA $\mathcal A=(Q, s, \delta, F)$. The order $ \le $ can now be \emph{partial} and (ii) is replaced with:
 (ii)'  $  (a_{1}=a_{2}  \wedge v_{1} < v_{2}) \rightarrow u_{1} \leq u_{2}$. 
 
Notice that (ii) and (ii)' are equivalent on \emph{total} orders. We measure the complexity of a co-lexicographic order by its width, so a Wheeler order is co-lexicographic order of width 1. 
  
In    \cite{NicNic2021}  it is   proved that the if an NFA admits a co-lexicographic order of width $p$, then:
\begin{enumerate}
    \item indexed pattern matching can always be solved in $\tilde O(\pi p^2)$ time (assuming constant-sized alphabet for simplicity);
    \item the standard powerset construction algorithm always produces an output whose size is exponentially bounded in $p$, rather than in the input's size;
    \item NFAs can be \emph{succinctly} encoded using $O(1+\log p)$ bits per edge  (assuming constant-sized alphabet for simplicity).
\end{enumerate}

Every automaton admits a co-lexicographic order \cite{NicNic2021}, so for every regular language we can define $ \text{width}^{N}(\mathcal L) $ ($\text{width}^{D}(\mathcal L) $) to be the minimum integer $ p $ such that there exists an NFA (DFA) accepting $ \mathcal{L} $ and admitting a co-lexicographic order of width $ p $. A language is Wheeler if and only if $ \text{width}^{N}(\mathcal L) = \text{width}^{D}(\mathcal L) = 1 $ \cite{DBLP:conf/soda/AlankoDPP20}. The following properties are proved in  a companion paper in preparation \cite{DBLP:journals/corr/abs-2102-06798}:
\begin{enumerate}
\item $ \text{width}^{N}(\mathcal L) $ and $\text{width}^{D}(\mathcal L) $ can be arbitrarily large;
\item There exist infinite languages $ \mathcal{L} $ such that $\text{width}^{D}(\mathcal L) > \text{width}^{N}(\mathcal L)$  (whereas in the Wheeler case equality always holds);
\item $\text{width}^N(\mathcal L)\leq  \text{width}^D(\mathcal L) \leq 2^{\text{width}^N(\mathcal L)} - 1 $ for every language $ \mathcal{L} $.
\end{enumerate}

Additionally, in the above mentioned   companion paper  we also  show that, for any given DFA, the best co-lexicographic order (that is, the one with minimum width) is the partial order that we  introduce below  in Definition \ref{def:det_automata_width}. 

In this paper we lift our point of view from automata to  a language-theoretic perspective,  
 tackling natural problems gravitating towards the following natural question: how do we determine the deterministic width of a language?

\section{The  width of a regular language: the Hasse automaton}\label{sec:ent}

 By identifying a state $ q $ with the set of strings $ I_q$ arriving at $q$ 
  we are  able to lift the co-lexicographic  order $\preceq$  to an order on the set of states of a DFA. This results in (partially) ordering states of a DFA by comparing strings arriving  at specific pairs of    states, as in the following definition. 
\begin{definition}\label{def:det_automata_width}
Let $ \mathcal A=(Q, s ,\delta, F) $ be a DFA. Let $ \preceq_\mathcal{A} $ be the partial order on $ Q $ such that for all $ q_{1}, q_{2}\in Q $ with $ q_1 \not = q_2 $: 
\begin{align*}
q_{1} \prec_{\mathcal A} q_{2} \iff I_{q_{1}} \prec    I_{q_{2}}.
\end{align*}
\end{definition}

An important measure of the complexity of a DFA $ \mathcal A$ is the ``distance" of the  partial order $\preceq_{\mathcal A}$ from a total order. Formally, such a distance is captured by the notion of width.  
\begin{definition}
The \emph{width} of a deterministic automaton, $\text {width}(\mathcal A)$,  is the width of the partial order $\preceq_{\mathcal A}$  or,  equivalently,  the maximum cardinality of an  antichain in $\preceq_{\mathcal A}$. 
\end{definition}

  We extend the notion of width from automata to languages, by considering  the best possible width among all automata recognizing the language.
  
\begin{definition}\label{def:det_width}
Given a regular language $\mathcal L$, its  \emph{deterministic width} is defined as follows:
\[\text {width}^D(\mathcal L)= \text{min} \{ \text{width}(\mathcal A)\mid ~ \mathcal A ~\text{is a DFA  } \wedge  \mathcal L(\mathcal A)=\mathcal L\}.\]
\end{definition}
Analogously, one may define the non-deterministic width. In this paper, we will only study the deterministic width of a language, so we shall  use the notation  $\text {width}(\mathcal L)$ instead of $\text {width}^D(\mathcal L)$. 

It would be convenient to have the width of a language equal to the width of its minimum automaton. Unfortunately, in general, this is not the case.

\begin{example}\label{ex:no_minimal}
In Figure \ref{fig:2minimumDFA}, the automaton  $ \mathcal{A}_1 $ on the left is a minimum DFA of width 3. The automata $ \mathcal{A}_2 $ and $ \mathcal{A}_3 $ in the center and on the right are two non-isomorphic automata of width 2 recognizing the same language recognized by $ \mathcal{A}_1 $ and having only one state more than $ \mathcal{A}_1 $. No DFA of width 1 recognizes this language. Hence, the width of the language is 2, the width of the minimum automaton is 3, and there is not a  unique minimum automaton among all DFAs of minimum width. More details can be found in Appendix \ref{appA}.
Although the concept of minimum automaton makes perfect sense for Wheeler languages \cite{DBLP:conf/soda/AlankoDPP20}, it may seem that no minimality result holds true for non-Wheeler languages. However, in section \ref{sec:minimal} we will explain why Example \ref{ex:no_minimal} is only seemingly undesirable, and we will derive an adequate notion of minimality.

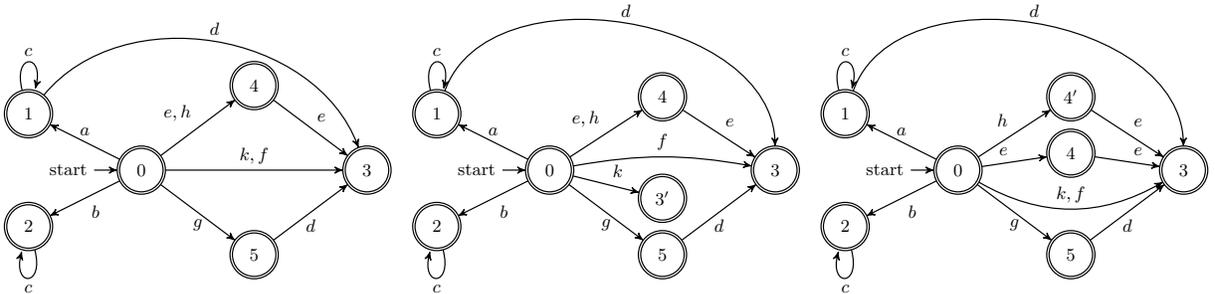
\begin{figure}[h!] 
 \begin{center}
 \scalebox{0.75}{ 
\begin{tikzpicture}[->,>=stealth', semithick, auto, scale=1]
\node[state,accepting,initial] (s)    at (0,0)	{$0$};

\node[state, accepting,label=above:{}] (1)    at (-2,1)	{$1$};
\node[state,accepting, label=above:{}] (2)    at (-2,-1)	{$2$};

\node[state, accepting,label=above:{}] (3)    at (4,0)	{$3$};
\node[state,accepting, label=above:{}] (4)    at (2,1.5)	{$4$};
\node[state,accepting, label=above:{}] (5)    at (2,-1.5)	{$5$};

\draw (s) edge [] node [above] {$a$} (1);
\draw (s) edge [] node [] {$b$} (2);

\draw (s) edge [] node [] {$k,f$} (3);
\draw (s) edge [] node [] {$e,h$} (4);
\draw (s) edge [] node [below] {$g$} (5);
\draw (5) edge [] node [below] {$d$} (3);
 \draw (4) edge [] node [] {$e$} (3);
 \draw (1) edge  [bend left=60, above] node [] {$d$} (3);
 \draw (1) edge  [loop above] node {$c$} (1);
 \draw (2) edge  [loop below] node {$c$} (2);
\end{tikzpicture}
}
 \scalebox{0.75}{ 
\begin{tikzpicture}[->,>=stealth', semithick, auto, scale=1]
\node[state,accepting,initial] (s)    at (-3,0)	{$0$};

\node[state, accepting,label=above:{}] (1)    at (-5,1)	{$1$};
\node[state,accepting, label=above:{}] (2)    at (-5,-1)	{$2$};

\node[state, accepting,label=above:{}] (3)    at (1,0)	{$3$};
 \node[state, accepting,label=above:{}] (3')    at (-1,-0.5)	{$3'$};
\node[state,accepting, label=above:{}] (4)    at (-1,1.3)	{$4$};
\node[state,accepting, label=above:{}] (5)    at (-1,-1.5)	{$5$};

\draw (s) edge [] node [above] {$a$} (1);
\draw (s) edge [] node [] {$b$} (2);

\draw (s) edge [] node [] {$k$} (3');
\draw (s) edge [bend left=10] node [] {$f$} (3);
\draw (s) edge [] node [] {$e,h$} (4);
\draw (s) edge [] node [below] {$g$} (5);
\draw (5) edge [] node [below] {$d$} (3);
 \draw (4) edge [] node [] {$e$} (3);
 \draw (1) edge  [bend left=80, above] node [] {$d$} (3);
 \draw (1) edge  [loop above] node {$c$} (1);
 \draw (2) edge  [loop below] node {$c$} (2);
\end{tikzpicture}
}
 \scalebox{0.75}{ 
\begin{tikzpicture}[->,>=stealth', semithick, auto, scale=1] 
 
 \node[state,accepting,initial] (sb)    at (4,0)	{$0$};

\node[state, accepting,label=above:{}] (1b)    at (2,1)	{$1$};
\node[state,accepting, label=above:{}] (2b)    at (2,-1)	{$2$};

\node[state, accepting,label=above:{}] (3b)    at (8,0)	{$3$};
 \node[state,accepting, label=above:{}] (4b')    at (6,1.3)	{$4'$};
\node[state,accepting, label=above:{}] (4b)    at (6,0.3)	{$4$};
\node[state,accepting, label=above:{}] (5b)    at (6,-1.5)	{$5$};

\draw (sb) edge [] node [above] {$a$} (1b);
\draw (sb) edge [] node [] {$b$} (2b);

\draw (sb) edge [bend right=30] node [] {$k, f$} (3b);
\draw (sb) edge [] node [] {$h$} (4b');
\draw (sb) edge [] node [] {$e$} (4b);
\draw (sb) edge [] node [below] {$g$} (5b);
\draw (5b) edge [] node [below] {$d$} (3b);
 \draw (4b) edge [] node [] {$e$} (3b);
  \draw (4b') edge [] node [] {$e$} (3b);
 \draw (1b) edge  [bend left=80, above] node [] {$d$} (3b);
 \draw (1b) edge  [loop above] node {$c$} (1b);
 \draw (2b) edge  [loop below] node {$c$} (2b);
\end{tikzpicture}
}
 \end{center}
 \caption{Three DFAs recognizing the same language.}
 \label{fig:2minimumDFA}.
\end{figure}
\end{example}

\vspace{-20pt}
We now exhibit a measure that on the minimum automaton will capture exactly the width of the accepted language: the  {\em entanglement number} of a DFA. 

 \begin{definition}\label{def:ent}
 Let $\mathcal B$ be a DFA with  set of  states $Q$. 
 \begin{enumerate}
     \item A subset $Q'\subseteq Q$ is {\em entangled} if there exists a monotone sequence  $(\alpha_i)_{i\in \mathbb N}$ such that for all $q' \in Q'$ it holds $\delta(s,\alpha_i)=q'$ for infinitely many $i$'s. In this case  the sequence $(\alpha_i)_{i\in \mathbb N}$ is said to be a {\em witness} for $Q'$. 
     \item A set $V\subseteq Pref (\mathcal{L(B)})$ is \emph{entangled}  if the set  $\{ \delta(s,\alpha)  ~|~ \alpha \in V\}$, consisting of all states occurring in $ V $, is entangled.
     
 \end{enumerate}
 Moreover, define: \begin{align*}
& {\text ent}(\mathcal B) = max\{|Q'| \mid Q'\subseteq Q \text{~and~} Q'\text{ is entangled } \}&\\
&  {\text ent}(\mathcal L) = {\text min} \{ {\text ent}(\mathcal B) ~|~ \mt B ~\text{is a DFA} ~ \land~\mt L(\mt B)=\mt L\}. &\\
\end{align*}
 \end{definition}

 \begin{remark}
Notice that any singleton $\{q\} \subseteq Q$  turns out to be entangled,  as  witnessed by the trivially monotone sequence $(\alpha_i)_{i\in \mathbb N}$ where all the $\alpha_i$'s are equal and $\delta(s,\alpha_i) = q$.
 \end{remark}
 
 For example, the entanglement of all  DFAs  in Figure  \ref{fig:2minimumDFA} is two,   because the only entangled subset of states is $\{1,2\}$. 
 
When two states $q\neq q'$ of a DFA  $ \mathcal{B} $  belong to an entangled set,   there are words $\alpha\prec \beta\prec \alpha'$ such that $\alpha,\alpha'\in I_q, \beta\in I_{q'}$, and so that neither $I_q\prec I_{q'}$ nor  $I_{q'}\prec I_{q}$ can hold. In other words,  two distinct states $q, q'$  belonging to an entangled set are always $\preceq_{\mt B}$-incomparable.  More generally,   the entanglement of a DFA is always smaller than  or equal to its width.
 \begin{lemma}\label{lem:ent_leq_width} 
  Let $ \mathcal B$ be a DFA. Then
  $\text {ent}(\mathcal B)\leq \text{width}(\mathcal B).$
 \end{lemma}
 
 The converse of the above inequality is not always true:   for the (minimum) DFA $\mt A_1 $ on the left of  Figure \ref{fig:2minimumDFA} we have  $\text {ent}(\mathcal A)=2,  \text{width}(\mathcal A)=3$.
 
We can easily prove that the    entanglement of a regular  language is realized by the minimum automaton accepting the language. 
 
 \begin{lemma}\label{lem:minent}
If  $\mathcal A $ is the minimum   DFA recognizing $\mathcal L$ then
 $\text{ent}(\mathcal A)= \text{ent}(\mathcal L)$.
 \end{lemma}

  Our aim is to prove  that   $\text{width} (\mt L(\mt A))= \text{ent} (\mt A)$ for $\mathcal A$ minimum. To prove this equality, we proceed as follows.  We first prove that, given a DFA $\mt B$,  there exists an equivalent DFA  $\mt B'$ such that   $\text{ent} (\mt B)= \text{width} (\mt B')$, that is, the automaton $\mt B'$ realizes the entanglement of $\mt B$ as its width. 
  
In order to give some intuition on the construction of the automaton $\mt B'$  we  use the {\em trace of a DFA},  that is the (in general) transfinite   sequence: $(\delta(s, \alpha))_{\alpha \in (\pf L,\preceq)}$ indexed over the totally ordered set $(\pf L,\preceq)$. We depict below a hypothetical $(\pf L,\preceq)$,  together with the  trace left by a  DFA $\mt B$ with set of states $\{q_1, q_2, q_3 \}$ and  $\delta(s, \alpha_i)=\delta(s, \alpha')=q_1$, 
  $\delta(s, \beta_i)= \delta(s, \beta_i')=q_2$, and  $\delta(s, \gamma_i)=q_3$: 
  \begin{center}
  \small
   \begin{tabular}{ccccccccccccccccccccccccccccccccccccccccccc}
&$\alpha_1$& $ \prec $ &$\beta_1$& $ \prec $ &$\alpha_2$& $ \prec $ &$\beta_2$& $ \prec $ & $ \ldots $ & $ \prec $ & $\alpha_i$& $ \prec $ &$\beta_i$& $ \prec $ & $ \ldots $ & $ \prec $ &$\beta_1'$ & $ \prec $ &$\gamma_1$& $ \prec $ &$\beta_2'$& $ \prec $ &$\gamma_2$& $ \prec $ & $ \ldots $ & $ \prec $ &$\beta_i'$& $ \prec $ &$\gamma_i$& $ \prec $ & $ \ldots $ & $ \prec $ & $\alpha'$&\\
 &$q_1$&&$q_2$&&$q_1$&&$q_2$&&$\ldots$&&$q_1$&&$q_2$&&$\ldots $&&$q_2$&&$q_3$&&$q_2$&&$q_3$&&$\ldots$&& $q_2$&&$q_3$&&$\ldots $&&$q_1$
\end{tabular}
\end{center}
   
Consider  the entanglement and  width of $\mt B$.  Notice that the sets $\{q_1, q_2\}$ and $\{q_2,q_3\}$ are entangled. The set  $\{q_1, q_3\}$ is not entangled  and therefore the set $\{q_1, q_2, q_3\}$ is not entangled.  However, $\{q_1, q_2, q_3\}$ contains pairwise incomparable states. 
 Hence the whole triplet $ \{q_1, q_2, q_3 \} $ does not contribute  to the entanglement but does contribute to the width  so that    $\text{ent}(\mt B)=2< \text{width}(\mt B)=3$. \\ 
In general,   an automaton where     incomparability  and  entanglement coincide  would have $\text{ent}(\mt B)=  \text{width}(\mt B)$. Hence, we would like   to force \emph{all} sets of incomparable states in $\mt B'$ to be entangled.  
 To this end, we will first slice $\pf L$ into  convex sets where incomparable states are entangled,  that is, 
   we will prove that there always exists a finite, ordered partition $\mt V=\{V_1, \ldots, V_r\}$ of 
$\pf L$ composed of  $\mt B$-entangled convex sets. In the    example above  we can write $\pf L=V_1\cup V_2\cup V_3 $, where: 
\[V_1=\{\alpha_1,\beta_1\ldots, \alpha_i, \beta_i, \ldots\}, V_2=\{\beta_1',\gamma_1, \ldots \beta_i', \gamma_i\ldots \},V_3=\{\alpha'\}\]
and the states occurring (and entangled) in $V_1,V_2,V_3$, respectively, are: $\{q_1,q_2\}$, $\{q_2,q_3\}$, and $\{q_1\}$. 
To eliminate  the pairwise incomparability  of $q_1,q_2,q_3$ we need to duplicate some of the original states. To this end, we will consider an equivalence relation  $\sim_{\mt V}^{\mt B}$ on $\pf L$ such that two strings are equivalent if and only if they are in the same $ I_q $ and all entangled $ V_i$'s between the two strings intersecate $I_q $. In the above example we have $\alpha_i \sim_{\mt V}^{\mt B} \alpha_j$  for all  integers $i,j$ but $\alpha_1\not \sim_{\mt V}^{\mt B} \alpha'$ because $\alpha_1\prec  V_2\prec \alpha'$ and $V_2\cap I_{q_1}=\emptyset$. 

In general, we will prove that the equivalence $\sim_{\mt V}^{\mt B}$    decomposes  the  set  of   words arriving in a state $q$ into a {\em finite} number of $\sim_{\mt V}^{\mt B}$-classes  and  induces a well-defined quotient automaton $ \mathcal{B'} $ equivalent to $ \mathcal{B} $.
It will turn out   that  if the $\mt B'$-states $ [q_1]_{\sim_{\mt V}^{\mt B}}, \ldots,  [q_k]_{\sim_{\mt V}^{\mt B}}$ are  pairwise $\preceq_{\mt B'}$-incomparable, then there exists $V\in \mt V$ such that $\{q_1, \ldots, q_k\}\subseteq V$; since $V$ is entangled  the set of states $ \{q_1, \ldots, q_k \} $  contribute  to the  entanglement number of $\mt B$ and we  will obtain $\text{width}({\mt B'})= \text{ent}({\mt B}) $.  

In our example, the new automaton $\mt B'$  will leave  the following trace:
  \begin{center}
  \small
   \begin{tabular}{ccccccccccccccccccccccccccccccccccccccccccc}
&$\alpha_1$&$ \prec $&$\beta_1$&$ \prec $&$\alpha_2$&$ \prec $&$\beta_2$& $ \prec $&$\ldots$& $ \prec $& $\alpha_i$& $ \prec $&$\beta_i$& $ \prec $&$\ldots$& $ \prec $&$\beta_1'$ & $ \prec $&$\gamma_1$& $ \prec $&$\beta_2'$& $ \prec $&$\gamma_2$& $ \prec $&$\ldots$& $ \prec $&$\beta_i'$& $ \prec $&$\gamma_i$& $ \prec $&$\ldots$& $ \prec $& $\alpha'$&\\
 &$q_1$&&$q_2$&&$q_1$&&$q_2$&&$\ldots$&&$q_1$&&$q_2$&&$\ldots$&&$q_2$&&$q_3$&&$q_2$&&$q_3$&&$\ldots$&& $q_2$&&$q_3$&&$\ldots$ &&$q_1'$
\end{tabular}
\end{center}

and     $\text{ent}(\mt B )=\text{width}(\mt B')=2$. 

\medskip

Formally, let us  start with the existence of  a finite decomposition $\mt V$. The following theorem will be a special case of a combinatorial property of convex subsets of an arbitrary linear order (see Theorem  \ref{thm:generalfdt} and Remark \ref{rem:special} in Appendix \ref{app:minimal_convex}). 

\begin{theorem}\label{thm:fdt} If $\mt B$ is a DFA then there exists a finite  partition $\mt V$  of $\pf {L(\mathcal{B})}$ whose elements  are  convex in $(\pf{L(\mathcal{B})}, \preceq)$ and entangled in $\mt B$.  
\end{theorem}

A partition $\mt V$ as  the one above is called an \emph{entangled, convex decomposition} of $\mt B$, or \emph{e.c. decomposition}, for short. If $ \mathcal{V} $ has also minimum cardinality, it is called a \emph{minimum-size} e.c. decomposition.  We are interested in minimum-size e.c. decompositions because they enforce  additional properties (see Remark \ref{rem: minimality} in Appendix \ref{app:minimal_convex}) implying that the relation that will be introduced in the following definition is right-invariant, so making it possible to define a quotient automaton. Note that in general a minimum-size e.c. decomposition of a DFA is not unique (see Example \ref{ex:no_uniqueness} in Appendix \ref{app:minimal_convex}).

\begin{definition}\label{def:sim}
Let $ \mathcal B$ be a DFA and let $\mt V$ be a minimum-size e.c. decomposition of $ \mathcal{B} $. The equivalence relation $\sim_{\mt V}^\mt B$  on  $ \pf {L(B)) }$  is defined by: 
\begin{align*}
    \alpha\sim_{\mt V}^\mt B \alpha' & \ \Leftrightarrow \   
  \delta(s,\alpha)=\delta(s,\alpha')  \wedge
  (\forall V \in \mt V) \;(\min \{\alpha, \alpha'\} \prec V \prec \max \{\alpha, \alpha'\} \rightarrow  V\cap I_{\delta(s,\alpha)}\neq \emptyset)
\end{align*}

When $ \mathcal{B} $ is clear from the context, we will simply write $ \sim_{\mt V} $.
\end{definition}

\begin{lemma}\label{lem:sim_finite}
Let $ \mathcal B$ be a DFA and let $\mt V$ be a minimum-size e.c. decomposition of $ Pref (\mathcal{L(B)})$. Then, $\sim_{\mt V}$ has a finite number of classes on  $\pf {L(\mt B)}$ and 
$\mt L(\mt B)$ is equal to the union of some $\sim_\mathcal{V}$-classes.
\end{lemma}

In order to prove that $\sim_{\mt V}$ is a right-invariant equivalence relation, we will first show a convenient consequence of considering $ \emph{minimum-size} $ e.c. decompositions: the relation $\sim_{\mt V}$ does not depend on the choice of $ \mathcal{V} $.

 \begin{lemma}\label{lem:easy_sim}
 Let $\mt B$ be a DFA and  let $\mt V$ be a minimum-size e.c. decomposition of $\mt B$. 
Then  $\alpha \sim_{\mt V} \alpha'$
    if and only if $\delta(s,\alpha)= \delta(s, \alpha')$ and   the interval $[\alpha, \alpha']^\pm$ is contained in a  finite union of convex, entangled sets  $C_1, \ldots,C_n$ in $\pf {L( \mt B)}$, with $C_i\cap I_{\delta(s,\alpha)}\neq \emptyset$ for all $i=1, \ldots,n$. In particular, $ \sim_{\mt V} $ is independent of the choice of $ \mathcal{V} $.
\end{lemma}

In view of the above lemma we shall drop the subscript $\mt V$ from $\sim_{\mt V}^\mt B $ and we will simply write $\sim^{\mt B} $  (or $ \sim $ when $ \mathcal{B} $ is clear from the context). We can now  prove that   $\sim$ is a right-invariant  equivalence relation.

\begin{lemma}\label{lem:simequivalence} Let $\mt B$ be a DFA. Then, the equivalence relation $\sim$ is right-invariant. 
\end{lemma}

We are now ready to complete the construction of the automaton $\mt B'$. The idea of the proof is to use $ \sim $ to build a quotient automaton.

  \begin{theorem}\label{thm:equivalentholeproof}
Let  $\mt B $ be  a DFA. Then, there exists a DFA $\mt B'$ such that $ \mathcal{L(B')} = \mathcal{L(B)} $ and $\text{ent}(\mt B)=\text{ent}(\mt B')= \text{width}(\mt B')$.
 \end{theorem}

In the remaining of this section, given a DFA $ \mathcal{B} $, we will always denote by $ \mathcal{B}' $ the DFA obtained by the construction in Theorem \ref{thm:equivalentholeproof}. Let us now apply  the construction to  the minimum automaton $\mt A$ of a regular language $\mt L$. 
\begin{definition} If $\mt A$ is the minimum automaton of a regular language $\mt L$,  the DFA $\mt A' $
is called the \emph{Hasse automaton} for $\mt L$ and it is denoted by $\mt H_{\mt L}$.
\end{definition} 

The name ``Hasse'' is justified by the fact that, as shown in the following theorem,  $\text{width}(\mt L)= \text{width}(\mt H_{\mt L})$ so that   the Hasse diagram of the partial order $\preceq_{\mt H}$ allows to ``visualize''  the width of the language.

 \begin{theorem}\label{thm:hasse} 
If $\mt A$ is the minimum automaton of the regular language $\mt L$, then:
\[\text{width}(\mt L)= \text{width}(\mt H_{\mt L})= \text{ent}(\mt A)=\text{ent}(\mt L).\]
 \end{theorem}

 The Hasse automaton $ \mt H_{\mt L} $ captures the width of a language. The following lemma shows that $ \mt H_{\mt L} $, which is defined starting from the minimum DFA $ \mathcal{A} $, also inherits a minimality property with respect to a natural class of automata. This reasonable result should be compared with Example \ref{ex:no_minimal}, which showed that the minimality problem does not admit a trivial interpretation when the concept of width is introduced. A satisfying solution will be provided in Section \ref{sec:minimal}.

 \begin{lemma}\label{lem:minHasse}
 Let $ \mathcal{L} $ be a language, and consider the class:
 \begin{equation*}
\mathscr{C} = \{\mathcal{B'} | \text{$ \mathcal{B} $ is a DFA and $ \mathcal{L(B)} = \mathcal{L} $} \}.
 \end{equation*}
 Then, there exists exactly one DFA being in $ \mathscr{C}$ and having the minimum number of states, namely, the Hasse automaton  $\mt H_{\mt L}$. In other words, $\mt H_{\mt L}$ is the minimum DFA of $ \mathscr{C} $.
 \end{lemma}

 \section{On computing the width of a regular language}\label{sec:computing}

The problem of computing the width of an automaton is in $ \texttt{P} $ if the automaton is deterministic, and it is $ \texttt{NP} $-hard if the automaton can be nondeterministic \cite{NicNic2021}. In this section we address the problem of determining the width of the \emph{language} recognized by an automaton. More precisely, we show that if we are given a regular language $ \mathcal{L} $ by means of \emph{any} DFA $\mathcal A$ accepting $ \mathcal{L} $, then the problem of computing $ p = width(\mathcal{L}) $ is in the class \texttt{XP}, that is, solvable in polynomial time for fixed values of $p$. We propose a dynamic programming algorithm that extends the ideas introduced in \cite{alanko2020wheeler} when solving the corresponding problem for Wheeler languages.

Theorem \ref{thm:hasse} suggests that the minimum automaton should contain all topological information required to compute the width of a language.
In the next theorem, we provide a graph-theoretical characterization of the width of a language based on the minimum automaton recognizing the language.

\begin{theorem}\label{thm:conditions width}
    Let $ \mathcal{L} $ be a regular language, and let $ \mathcal{A} $ the minimum DFA of $ \mathcal{L} $, with set of states $ Q $. Let $ k \ge 2 $ be an integer. Then, $ width(\mathcal{L}) \ge k $ if and only if there exist strings $ \mu_1, \dots, \mu_k $ and $ \gamma $ and there exist pairwise distinct $ u_1, \dots, u_k \in Q $ such that:
    \begin{enumerate}
        \item $ \mu_j $ labels a path from the initial state $ s $ to $ u_j $, for every $ j = 1, \dots, k $;
        \item $ \gamma $ labels a cycle starting (and ending) at $ u_j $, for every $ j = 1, \dots, k $;
        \item either all the $ \mu_j $'s are smaller than $ \gamma $ or $ \gamma $ is smaller than all $ \mu_j $'s;
        \item $ \gamma $ is not a suffix of $ \mu_j $, for every $ j = 1, \dots, k $.
    \end{enumerate}
\end{theorem}

In Appendix \ref{app:computing} we prove further results to bound the lengths of $ \mu_1, \dots, \mu_k $ and $ \gamma $, so making it possible to bound the running time of a dynamic programming algorithm based on Theorem \ref{thm:conditions width}. We can then conclude:

  \begin{theorem}\label{thm:dyn prog}
 Let $\mathcal L$ be  a regular language, given as input by means of any DFA $\mathcal A = (Q, s, \delta,F)$ recognizing $ \mathcal{L}$. 
Then, $p = \text{width}(\mathcal L)$ is  computable in time $|Q|^{O(p)}$.  
  \end{theorem}

\section{The convex Myhill-Nerode theorem}\label{sec:minimal}

Let $ \mathcal{L} \subseteq \Sigma^* $ be a language. The \emph{Myhill-Nerode equivalence} for $ \mathcal{L} $ is the right-invariant equivalence relation $ \equiv_\mathcal{L} $ on $ \pf L $ such that for every $ \alpha, \beta \in \pf L $ it holds:
\begin{equation*}
    \alpha \equiv_\mathcal{L} \beta \iff \{\gamma \in \Sigma^* | \alpha \gamma \in \mathcal{L}  \} = \{\gamma \in \Sigma^* | \beta \gamma \in \mathcal{L}  \}.
\end{equation*}

In this paper, we have described the hierarchy of regular languages by means of their widths, thus the natural question is whether a corresponding Myhill-Nerode theorem can be provided. For instance, given a regular languages $ \mathcal{L} $, if we consider all DFAs recognizing $ \mathcal{L} $ and having width equal to $ width (\mathcal{L}) $, is there a unique such DFA having the minimum number of states? In general, the answer is "no", as showed in Example \ref{ex:no_minimal}.

The non-uniqueness can be explained as follows. If a DFA of width $ p $ recognizes $ \mathcal{L} $, then $ \pf L $ can be partitioned into $ p $ sets, each of which consists of the (disjoint) union of some pairwise comparable $ I_q $'s.  However, in general the partition into $ p $ sets is not unique, so it may happen that two distinct partitions lead to two non-isomorphic minimal DFAs with the same number of states.  For example, in Figure \ref{fig:2minimumDFA}, the chain partition $ \{ \{0,1,4\}, \{2,3, 5, 3'\} \} $ of the DFA in the center induces the partition $ \{ac^*  \cup \{\epsilon, e, h \}, bc^* \cup ac^*d \cup \{gd, ee, he,f,k,g \} \} $ of $ Pref (\mathcal{L}) $, whereas the chain partition $ \{ \{0,1,3\}, \{2,4, 5, 4'\}\} $ of the DFA on the right induces the partition $\{ac^* \cup ac^*d \cup \{\epsilon, gd,  ee,he, f,k \}, bc^* \cup \{e, h, g\}\}$ of $ Pref (\mathcal{L})$.

As a consequence, no uniqueness results can be ensured as long as partitions are not fixed. But what happens if we fix a partition? As we will prove in this section, once a partition is fixed, it is possible to prove a full Myhill-Nerode theorem, so providing a DFA-free characterization of convex languages and a minimum DFA.

More formally, let $ \mathcal{A} = (Q, s, \delta, F) $ be a DFA, and let $ \{Q_i \}_{i = 1}^p $ be a $ \le_\mathcal{A} $-chain partition of $ Q $. For every $ i \in \{1, \dots, p \} $, define:
\begin{equation*}
    Pref(\mathcal{L(A)})^i = \{\alpha \in Pref(\mathcal{L(A)}) | \delta(s, \alpha) \in Q_i  \}.
\end{equation*}
Then $ \{Pref(\mathcal{L(A)} \}_{i = 1}^p $ is a partition of $ Pref (\mathcal{L(A)}) $, and from now on we will think of such a partition as fixed. We now consider the class of all DFAs accepting $ \mathcal{L} $ and inducing the considered partition.

\begin{definition}\label{def:psortableNFA}
Let $ \mathcal{A} = (Q, s, \delta, F) $ be an DFA, and let $ \mathcal{P} = \{U_1, \dots, U_p \} $ be a partition of $ Pref (\mathcal{L(A)}) $. We say that $ \mathcal{A} $ is $ \mathcal{P} $-sortable if there exists a $ \le_\mathcal{A} $-chain partition $ \{Q_i \}_{i = 1}^p $ such that for every $ i \in \{1, \dots, p \} $:
\begin{equation*}
    Pref(\mathcal{L(A)})^i = U_i.
\end{equation*}
\end{definition}

We wish to give a DFA-free characterization of languages $ \mathcal{L} $ and the partitions $ \mathcal{P} $ of $ \pf L $ for which there exists a $ \mathcal{P} $-sortable DFA. As in the Myhill-Nerode theorem, we aim to determine which properties an equivalence relation $ \sim $ should satisfy to ensure that a canonical construction provides a $ \mathcal{P} $-sortable DFA. First, $ \mathcal{L} $ must be regular, so $ \sim $ is expected to be right-invariant. In order to develop some intuition on the required properties, let us consider an equivalence relation which plays a key role in the classical Myhill-Nerode theorem. Let $ \mathcal{A} = (Q, s, \delta, F) $ be a $ \mathcal{P} $-sortable DFA, and let $ \sim_\mathcal{A} $ the equivalence relation on $ Pref (\mathcal{L(A)})$ whose equivalence classes are $ \{I_q | q \in Q \} $.   Notice that equivalent string  end up in the same state and so in the same element of   $ \mathcal{P}$ (\emph{$ \mathcal{P} $-consistency}), and since all states in each $ \preceq_\mathcal{A} $-chain $ Q_i $ are comparable, then each $ I_q $ must be convex in the corresponding element of $ \mathcal{P} $ (\emph{$ \mathcal{P} $-convexity}). Formally:
 
\begin{definition}\label{def:PconsistentPconvex}
Let $ \mathcal{L} \subseteq \Sigma^* $ be a language, and let $ \sim $ be an equivalence relation on $ \pf L $. Let $ \mathcal{P} = \{U_1, \dots, U_p \} $ be a partition of $ \pf L $.
\begin{enumerate}
\item For every $ \alpha \in \pf L $, let $ U_\alpha $ be the unique element $ U_i $ of $ \mathcal{P} $ such that $ \alpha \in U_i $.
\item We say that $ \sim $ is \emph{$ \mathcal{P} $-consistent} if for every $ \alpha, \beta \in \pf L $, if $ \alpha \sim \beta $, then $ U_\alpha = U_\beta $.
\item Assume that $ \sim $ is $ \mathcal{P} $-consistent. We say that $ \sim $ is \emph{$ \mathcal{P} $-convex} if for every $ \alpha \in \pf L $ we have that $ [\alpha]_{\sim} $ is a convex set in $ (U_\alpha, \preceq) $.
\end{enumerate}
\end{definition}

As we will prove, these are exactly the required property for a DFA-free characterization. Lastly, we wish to define an equivalence relation with the above properties inducing a (unique) minimum $ \mathcal{P} $-sortable DFA. We expect such relation to be related with the Myhill-Nerode equivalence, so a sensible choice is $ \equiv_\mathcal{L}^\mathcal{P} $, which is the coarsest $ \mathcal{P} $-consistent, $ \mathcal{P} $-convex and right-invariant equivalence relation refining $ \equiv_\mathcal{L} $. The proof that such a coarsest equivalence relation exists is technical and it is provided in the Appendix \ref{app:nerode} (see Corollary \ref{cor:coarsest}). We now have all the required definitions to state our Myhill-Nerode theorem, which generalizes the one for Wheeler languages \cite{DBLP:conf/soda/AlankoDPP20}.
\begin{theorem}[Convex Myhill-Nerode theorem]\label{thm:MN}
    Let $ \mathcal{L} $ be a language. Let $ \mathcal{P} $ be a partition of $ \pf L $. The following are equivalent:
    \begin{enumerate}
        \item $ \mathcal{L} $ is recognized by a $ \mathcal{P} $-sortable DFA.
        \item $ \equiv_\mathcal{L}^\mathcal{P} $ has finite index.
        \item $ \mathcal{L} $ is the union of some classes of a $ \mathcal{P} $-consistent, $ \mathcal{P} $-convex, right invariant equivalence relation on $ \pf L $ of finite index.
    \end{enumerate}
Moreover, if one of the above statements is true (and so all the above statements are true), then there exists a unique minimum $ \mathcal{P} $-sortable DFA (that is, two $ \mathcal{P} $-sortable DFA having the minimum number of states must be isomorphic).
\end{theorem}

Notice that for a language $ \mathcal{L}$  it holds $ width(\mathcal{L}) = p $ if and only if there exists a partition $ \mathcal{P} $ of $ \pf L $ that satisfies any of the statements in Theorem \ref{thm:MN}.

Given a $ \mathcal{P} $-sortable DFA recognizing $ \mathcal{L} $, it is possible to build the minimum $ \mathcal{P} $-sortable DFA recognizing $ \mathcal{L} $ in polynomial time by generalizing the algorithm for the Wheeler case \cite{DBLP:conf/soda/AlankoDPP20}. More details on the algorithm and its complexity will follow in a companion paper.  
\section{Conclusions}

The two topics discussed in this paper are \emph{regular languages} and \emph{order}. The link between these two important notions is built upon a technique to cast the co-lexicographic order of (sets of) strings over a (partial) order of the states of a finite state automaton accepting a regular language.
Finite state automata are among the most basic models of computation and ordering is among the most basic data-structuring mechanisms.  Hence, it comes as no surprise that the most powerful algorithmic techniques for solving pattern matching obtained over the last decades exploit precisely this notion of ordering.

Proceeding along the same line, the main theorem that we proved here is the existence of a canonical, partially-ordered automaton realising the minimum size of an anti-chain (the ``width") of its states. The partial order of this canonical automaton is built from a convex decomposition of the co-lexicographically ordered collection of the prefixes of the language. This suggests an \emph{absolute} measure of complexity for regular languages: the width of its canonical automaton. Such a measure turns out to be somehow orthogonal to (and sometimes conflicting with) the mere counting of the number of states of the automaton. 

 At the same time, the width has a clear algorithmic interpretation because it measures the disposition of the language to be ``index-able'' \cite{NicNic2021}. 

Several intriguing questions can be raised at this point: how much of our findings can be extended to other models of computation---starting from non-deterministic finite automata? What is the exact role played by the specific (co-lex) ordering that we used? Are there alternatives? Is it possible to give effective or approximate constructions of the Hasse automaton? A positive answer to this last question would go a long way towards building a minimum-size automaton among those of minimum width accepting a regular language.

Finally, it would be remarkable to determine a regular-expression-like description of the hierarchy of regular languages based on the width, which would imply both theoretical results (a Kleene-like theorem) and concrete applications (regex pattern matching). 

 \appendix 
 \section{Proofs of Section \ref{sec:ent}}\label{appA}
 
Let us provide more details about Example \ref{ex:no_minimal}.

\paragraph*{Example \ref{ex:no_minimal}.}
 Consider the DFA $ \mathcal{A}_1 $ on the left of Figure  \ref{fig:2minimumDFA} and let $ \mathcal{L} = \mathcal{L}(\mathcal{A}_1) $. Then, $ \mathcal{A}_1 $ is a minimum DFA with states $0, \ldots,5$ and $I_0=\{\epsilon\}$, $I_1=\{ac^*\}$, $I_2=\{bc^*\}$, $I_3=\{ac^*d, gd,  ee,he, f,k \}$, $I_4=\{e,h\}$, $I_5=\{g\}$. States $1 $ and $ 2 $ are $\preceq_{\mt A}$-incomparable because $a\in I_1, b\in I_2, ac\in I_1$ and $a\prec b\prec ac$.
Similarly one checks that states $3,4,5$ are pairwise  $\preceq_{\mt A}$-incomparable. On the other hand, $0$ is the minimum and states $1,2$ precede states $3,4,5 $ in the order $\preceq_{\mt A}$. We conclude:
\begin{equation*}
\preceq_{\mt A}=\{(0,1),(0,2), (0,3), (0,4), (0,5),(1,3),(1,4),(1,5),(2,3),(2,4),(2,5)\} \cup \{(i, i) \;|\; 0 \le i \le 5 \}\}.
\end{equation*}
The width of the DFA is $ 3 $  because $ \{3, 4, 5 \} $ is the largest $ \preceq_\mathcal{A} $-antichain. A $ \preceq_\mathcal{A} $-chain partition of cardinality $ 3 $ is, for example, $ \{\{0, 1, 3 \}, \{2, 4 \}, \{5 \} \} $.

Let us prove that $ width(\mathcal{L}) \ge 2 $. Suppose by contradiction that there exists a   DFA $\mt B$ of width $1$    recognizing  $\mt L $. Then, the order $\preceq_{\mt B}$ is  total. Moreover,  there   exists a state $q$ such that two words of the infinite set 
   $ac^*\in \pf L$, say $ac^i, ac^j$  with $i<j$,    belong to $I_q$.
   Since $ac^i\prec bc^i\prec ac^{i+1}$ with 
   $bc^i\in \pf L$ and 
   $bc^i\not \equiv_{\mt L} ac^{i+1}$ it follows that $bc^i\not \in I_q$. If $q'$ is such that $bc^i\in I_{q'}$ we have
   that $q$ and $q'$ are $\preceq_{\mt B}$-incomparable, a contradiction.

Finally, let $ \mathcal{A}_2 $ be the DFA in the center of Figure  \ref{fig:2minimumDFA} and let $ \mathcal{A}_3 $ be the DFA on the right of Figure  \ref{fig:2minimumDFA}. Notice that $ \mathcal{L}(\mathcal{A}_2) = \mathcal{L}(\mathcal{A}_3) = \mathcal{L} $, and $ \mathcal{A}_2 $ and $ \mathcal{A}_3 $ have just one more state than $ \mathcal{A}_1 $ and are non-isomorphic. We know that $ \mathcal{A}_2 $ and $ \mathcal{A}_3 $ cannot have width equal to 1. On the other hand, they both have width 2, as witnessed by $\{ \{0,1,4\}, \{2,3 ,5, 3' \} \}$ (for $ \mathcal{A}_2 $) and $ \{ \{0,1,3\}, \{2,4, 5, 4' \}\} $ (for $ \mathcal{A}_3 $).

\paragraph*{}

Here is the proof of Lemma \ref{lem:minent}.

\paragraph*{\textbf{Statement of Lemma \ref{lem:minent}.}}
If  $\mathcal A $ is the minimum   DFA recognizing $\mathcal L$ then
 $\text{ent}(\mathcal A)= \text{ent}(\mathcal L)$.

\paragraph*{\textbf{Proof}}
We prove that $\text{ent}(\mathcal A)\leq \text{ent}(\mathcal B)$, for any  $\mathcal B$ equivalent to $\mathcal A $. Suppose $q_1, \ldots, q_k$ are pairwise distinct states which are    entangled in $\mathcal A $, witnessed by the monotone sequence $(\alpha_i)_{i\in \mathbb N}$. Since $\mathcal A$ is minimum, each  $I_{q_j}$ is a union of a finite number of  $I_u$, with $u\in Q_{\mt B}$. 
 The monotone sequence $(\alpha_i)_{i\in \mathbb N}$
 goes through $q_j$ infinitely often,  so there must be a state $u_j\in Q_{\mt B}$ such that $I_{u_j}\subseteq I_{q_j}$ and  $(\alpha_i)_{i\in \mathbb N}$
 goes through $u_j$ infinitely often.
Then  $(\alpha_i)_{i\in \mathbb N}$ goes through the pairwise distinct states $u_1,\ldots, u_k$ infinitely often  and $u_1, \ldots, u_k$  are entangled in  $\mathcal B$. \qed

\subsection{The minimum-size entangled convex  decomposition}\label{app:minimal_convex}

This section is devoted to the proof of Theorem \ref{thm:fdt}, which will follow from a general result valid for an arbitrary linear order. From now on, we fix a linear order $ (Z, \le) $ and a finite partition $ \mathcal P = \{ P_{1}, \ldots , P_{m}\} $ of $ Z $. Theorem \ref{thm:fdt} will follow by letting $ (Z, \le) = (Pref(\mathcal{L(B)}, \preceq) $ and $ \mathcal{P} = \{I_q | q \in Q \} $.

It is convenient to think of  $\mathcal P$ as an alphabet from which we can generate finite or infinite strings. A finite string $P_{ 1}\ldots P_{ k} \in \mathcal P^{*}$  is said to be {\em generated} by  $X \subseteq Z $, if there exists a  sequence $x_{1}\leq\ldots  \leq x_{k}$  of elements in $X$   such that 
$x_{j}\in P_{j}$, for all $j=1, \ldots,k$.   
In these hypotheses, we also say that $P_{ 1}\ldots P_{ k} $ {\em occurs} in $ X $  at  $x_{1}, \ldots, x_{k}$.
Similarly,  an infinite string  $P_{1}\ldots P_{k} \ldots   \in \mathcal P^{\omega}$ is  generated by  $X \subseteq Z $ if there exists  a monotone  sequence  $(x_i)_{i\in \mathbb N}$  of elements in $X$   such that 
$x_{j}\in P_{j}$, for all  $j\in \mathbb N$.

We can now generalize Definition \ref{def:ent}.

\begin{definition}\label{def:entangled}
Let $ (Z, \le) $ be a total order, and let $ \mathcal{P} $ be a partition of $ Z $. Let $ X \subseteq Z $.
\begin{enumerate}
    \item We define $\mathcal P_{X}=\{P\in \mathcal P : P \cap X\neq\emptyset\}$.
    \item If $\mathcal P'=\{P_1, \ldots, P_m\}\subseteq \mathcal P$, we say that $\mathcal P'$ is  {\em entangled} in $X$   if the infinite string $(P_1, \ldots, P_m)^\omega$ is generated by $X$.
    \item We say that $ X $ is {\em entangled} if $ \mathcal P_{X}$ is entangled in $ X $.
\end{enumerate}
\end{definition}

The property of being entangled is captured by an infinite string $(P_1, \ldots, P_m)^\omega$. Let us prove that, in fact, arbitrary long strings express the same property.

\begin{lemma}\label{lem:m_implies_wm}
Let $ (Z, \le) $ be a total order, and let $ \mathcal{P} $ be a partition of $ Z $. Let $ X \subseteq Z $, and let $ \mathcal P'=\{P_1, \ldots, P_m\}\subseteq \mathcal P$. The following are equivalent:
\begin{enumerate}
    \item For every $k\in \mathbb N$, the string $ (P_1, \ldots, P_m)^k$ is generated by $X$.
    \item $ (P_1, \ldots, P_m)^\omega $ is generated by $ X $
\end{enumerate}
\end{lemma}

\begin{proof}
The nontrivial implication is $ (1) \rightarrow (2) $. If $ m = 1 $ the conclusion is immediate, so we can assume $ m \ge 2 $. Let  $ \alpha_k = (x^k_1, \dots, x^k_{mk}) $ be an increasing sequence witnessing that $ (P_{1} \cdots P_{m})^{k} $ occurs in $ X $. 
We distinguish two cases:
\begin{enumerate}
    \item There exists an integer $ i_0 $ such that for every integer $ i > i_0 $ and for every integer $ r $ such that $ x_i^r $ is defined there exists an integer $ h(i, r) $ such that $ x^{h(i, r)}_{i + 1} $ is defined and $ x^{h(i, r)}_{i + 1} < x^{r}_{i} $. Let $ i > i_0 $ be such that $ x_i^r \in P_1 $ for every $ r $ such that $ x_i^r $ is defined, fix any such $ r $, and define $ r_{i} := r $. Now, for $ j > i $, define recursively $ r_{j + 1} := h(j, r_j) $. Then, the sequence $ x_{i}^{r_i} $, $ x_{i + 1}^{r_{i + 1}}$, $ x_{i + 2}^{r_{i + 2}} $, $ \dots $ is a decreasing sequence in $ X $ witnessing that $ (P_{1} \cdots P_{m})^{\omega} $ is generated by $ X $.
    \item For every integer $ i_0 $ there exists an integer $ i > i_0 $ and there exists an integer $ r(i_0) $ such that $ x^{r(i_0)}_{i} $ is defined and for every integer $ h $ for which $ x^{h}_{i + 1} $ is defined it holds $ x^{h}_{i + 1} > x^{r(i_0)}_{i} $ (equality cannot hold because $ m \ge 2$ ). Now, fix $ i_0 = 1 $ and call $ i_1 $ the correspondent $ i > i_0 $. Next, consider $ i_1 $ and call $ i_2 $ the correspondent $ i > i_1 $, and so on. Hence $ i_1 < i_2 < i_3 < \dots $. Consider the sequence $ x_1^{r(i_0)} $, $ x_2^{r(i_0)} $, $ \dots $, $ x_{i_1 - 1}^{r(i_0)} $, $ x_{i_1}^{r(i_0)} $, $ x_{i_1 + 1}^{r(i_1)} $, $ x_{i_1 + 2}^{r(i_1)} $, $ \dots $, $ x_{i_2 -1 }^{r(i_1)} $, $ x_{i_2}^{r(i_1)} $, $ x_{i_2 + 1}^{r(i_2)} $, $ x_{i_2 + 2}^{r(i_2)} $, $ \dots $, $ x_{i_3 -1 }^{r(i_2)} $, $ x_{i_3}^{r(i_2)} $, $ x_{i_3 + 1}^{r(i_3)} $, $ \dots $. We claim that this is an increasing sequence in $ X $ witnessing that $ (P_{1} \cdots P_{m})^{\omega} $ is generated by $ X $. First, notice that all $ x_i^k $'s in the sequence are defined (that is, it actually holds $ i \le mk $ ) because we know that $ x_{i_1}^{r(i_0)} $, $ x_{i_2}^{r(i_1)} $, $ \dots $ are defined. Moreover, since the subscripts of the sequence are $ 1 $, $ 2 $, $ 3 $, $ \dots $, then the sequence witnesses that $ (P_1, \dots, P_m)^\omega $ occurs in $ X $, if we prove that it is monotone. Let us prove that this sequence is increasing. The subsequence $ x_1^{r(i_0)} $, $ x_2^{r(i_0)} $, $ \dots $, $ x_{i_1 - 1}^{r(i_0)} $, $ x_{i_1}^{r(i_0)} $ is increasing because it is contained in $ \alpha_{r(i_0)} $, the subsequence $ x_{i_1 + 1}^{r(i_1)} $, $ x_{i_1 + 2}^{r(i_1)} $, $ \dots $, $ x_{i_2 -1 }^{r(i_1)} $, $ x_{i_2}^{r(i_1)} $ is increasing because it is contained in $ \alpha_{r(i_1)} $, and so on. Finally, $ x_{i_1}^{r(i_0)} < x_{i_1 + 1}^{r(i_1)} $, $ x_{i_2}^{r(i_1)} < x_{i_2 + 1}^{r(i_2)} $ and so on by the definition of $ i_1 $, $ i_2 $, $ \dots $. \qed
\end{enumerate}
\end{proof}

Let us generalize the definition of e.c. decomposition given in Section \ref{sec:ent} for DFAs.

\begin{definition}
Let $ (Z, \le) $ be a total order, and let $ \mathcal{P} $ a partition of $ Z $. We say that a partition  $ \mathcal V  $ of $Z$ is an {\em e.c. decomposition} of $ \mathcal P$ in $(Z,\leq)$ if
all elements of $ \mathcal V  $ are entangled and convex in $ (Z, \le) $.
\end{definition}

 \begin{example}
 Consider $ (\mathbb{Z}, \le) $, where $\mathbb Z$ is the set of all integers and $ \le $ is the usual order on $\mathbb Z$. Let $\mathcal P=\{P_1,P_2,P_3\}$ the partition of $ \mathbb{Z} $ defined as follows:
 
   \begin{align*}
& P_1=\{n \leq 0 : n \text{~ is odd}\} \cup \{n >  0 : n \equiv 1 ~\text{mod } 3\}&\\
&  P_2=\{n\leq 0 : n \text{~ is even}\}  \cup \{n >  0 : n \equiv 2 ~\text{mod } 3\} , &\\
 &P_3=\{n >  0 : n \equiv 0 ~\text{mod } 3\} &
\end{align*}

The partition $\mathcal P$ generate the following \emph{trace}  over $\mathbb Z$:
\[\ldots P_1P_2P_1P_2 \ldots P_1P_2 P_1P_2P_3P_1P_2P_3\ldots\]

Now define $\mt V=\{V_1,V_2\}$, where 
$V_1=\{n\in \mathcal Z : n\leq 0\}$, $V_2=\{n\in \mathcal Z : n>0\}$. It is immediate to check that $ \mathcal{V} $ is an e.c. decomposition of $ \mathcal{P} $ in $ (\mathbb{Z}, \le) $. More trivially, even $\mt V'=\{\mathbb Z\}$ is an e.c. decomposition of $\mathcal P$ in $ (\mathbb{Z}, \le) $.
 \end{example}

Let $ (Z, \le) $ be a total order, and let $ \mathcal{P} $ a partition of $ Z $. Notice that $ \mathcal{V} = \{\{z\} | z \in Z \} $ is an e.c. decomposition of $ \mathcal{P} $ in $ (\mathbb{Z}, \le) $. However, if $ Z $ is an infinite set, then $ \mathcal{V} $ is an infinite e.c. decomposition. The main result of this section is that if $ \mathcal{P} $ is a finite partition, then there exists a \emph{finite} e.c. decomposition of $ \mathcal{P} $ in $ (\mathbb{Z}, \le) $ even when $ Z $ is an infinite set.

\begin{theorem} \label{thm:generalfdt}
    Let $ (Z, \le) $ be a total order, and let $ \mathcal{P} = \{P_1, \dots, P_m \} $ be a finite partition of $ Z $. Then, $ \mathcal{P} $ admits a finite e.c. decomposition in $ (Z, \le) $.
\end{theorem}

\begin{proof}
We proceed by induction on $ m = |\mathcal{P} | $. If $ m = 1 $, then $ \mathcal{P} = \{Z \} $, so $ \{Z \} $ is an e.c. decomposition of $ \mathcal{P} $ in $ (Z, \le) $. Now assume $ m \ge 2 $. If $ (P_1, \dots, P_m)^\omega $ is generated by $ Z $, then again $ \{Z \} $ is an e.c. decomposition of $ \mathcal{P} $ in $ (Z, \le) $. Otherwise, let $ \pi $ any permutation of the set $ \{1, \dots, m \} $. Since $ (P_1, \dots, P_m)^\omega $ is not generated by $ Z $, then $ (P_{\pi (1)}, \dots, P_{\pi (m)})^\omega $ is not generated by $ Z $, hence by Lemma \ref{lem:m_implies_wm} there exists an integer $ s_\pi $ such that $ (P_{\pi (1)}, \dots, P_{\pi (m)})^{s_\pi}$ is not generated by $ Z $. Now, consider the following procedure:
\begin{algorithmic}[1]
    \STATE{$t \leftarrow 1; Z_{1} \leftarrow Z; \mathcal V \leftarrow \{Z_{1}\}$;}
    \COMMENT{initialise the partition}
    \FOR{ $ \pi $ being a permutation of $ \{1, \dots, m \}$}
          \WHILE{ $\exists Z_{i} \in \mathcal V$ generating $ (P_{\pi(1)}\cdots P_{\pi(m)})^{2} $ }
            \STATE{let $ \alpha_{1}  < \cdots < \alpha_{m}  < \alpha_{1}' < \cdots < \alpha_{m}' $ in $ Z_{i} $ be such that $\alpha_{j}, \alpha_{j}' \in P_{\pi(j)}$;   }
            \STATE{$   Z_{i} \leftarrow \{ \alpha \in Z_{i}\ | \ \alpha \leq \alpha_{m}\}$;}
            \STATE{$   Z_{t+1} \leftarrow \{ \alpha \in Z_{i}\ | \ \alpha > \alpha_{m}\}$;}
            \STATE{$ \mathcal V \leftarrow \mathcal V \cup \{Z_{t+1}\} $; }
            \STATE{$ t \leftarrow t+1 $; }
          \ENDWHILE
    \ENDFOR
            \RETURN{$\mathcal V$}
\end{algorithmic}
The procedure starts from $ \mathcal{V} = \{Z \} $, and recursively partitions one element from $ \mathcal{V} $ into two nonempty convex subsets as long as $ (P_{\pi (1)}, \dots, P_{\pi (m)})^2 $ occurs in the considered element. Notice the procedure ends after at most $ \sum_{\pi} s_\pi $ iterations, returning a finite partition $ \mathcal{V} $ of $ Z $ into convex sets such that every $ V \in \mathcal{V} $ has the property that for every permutation $ \pi $ the string $ (P_{\pi (1)}, \dots, P_{\pi (m)})^2 $ does not occur in $ V $. It will suffice to prove that for every $ V \in \mathcal{V} $ the partition $ \{P \cap V |P \in \mathcal{P}_V \} $ of $ V $ (whose cardinality is at most $ m $) admits a finite e.c. decomposition in $ (V, \le) $, because then we will obtain a finite e.c. decomposition of $ \mathcal{P} $ in $ (Z, \le) $ by merging the decompositions obtained for each $ V $.

Fix $ V \in \mathcal{V} $. If $ |\{P \cap V |P \in \mathcal{P}_V \}| < m $, the conclusion follows by the inductive hypothesis. Now assume $ |\{P \cap V |P \in \mathcal{P}_V \}| = m $. It will suffice to prove the following: if $ (Z, \le) $ is a total order, and $ \mathcal{P} = \{P_1, \dots, P_m \} $ is a finite partition of $ Z $ such that for every permutation $ \pi $ of $ \{1, \dots, m \} $ the string $ (P_{\pi (1)}, \dots, P_{\pi (m)})^2 $ does not occur in $ Z $, then $ \mathcal{P} $ admits a finite e.c. decomposition in $ (Z, \le) $.

Let $ k \ge 1 $ be the number of distinct permutations $ \pi $ of $ \{1, \dots, m \} $ such that $ (P_{\pi (1)}, \dots, P_{\pi (m)}) $ occurs in $ Z $. We proceed by induction on $ k $. If $ k = 1 $, then each set $ Z \cap P_j $ is convex and trivially entangled, so $ \{Z \cap P_j | 1 \le j \le m \} $ is an e.c. decomposition of $ \mathcal{P} $ in $ (Z, \le) $. Now assume $ k \ge 2 $. Let $ \pi_0 $ be a permutation such that $ (P_{\pi_0 (1)}, \dots, P_{\pi_0 (m)}) $ occurs in $ Z $. Define:
\begin{equation*}
    Z_1 = \{\alpha \in Z | \text{$ \exists \alpha_1, \dots, \alpha_m $ with $ \alpha \le \alpha_1 < \dots \alpha_m $ and $ \alpha_i \in P_{\pi_0 (i)}$ \}}
\end{equation*}
and $ Z_2 = Z \setminus Z_1 $. Notice that $ Z_1 $ is nonempty because $ (P_{\pi_0 (1)}, \dots, P_{\pi_0 (m)}) $ occurs in $ Z $. Let us prove that $ (P_{\pi_0 (1)}, \dots, P_{\pi_0 (m)})) $ does not occur in $ Z_1 $ nor in $ Z_2 $, and that $ Z_2 $ is nonempty. Just observe that if $ \alpha_1, \dots, \dots \alpha_m $ is a witness for $ (P_{\pi_0 (1)}, \dots, P_{\pi_0 (m)}) $ in $ Z $, then $ \alpha_1 \in Z_1 $, but $ \alpha_m \in Z_2 $, otherwise $ (P_{\pi_0 (1)}, \dots, P_{\pi_0 (m)})^2 $ would occur in $ Z $. Moreover $ Z_1 $ and $ Z_2 $ are convex. It will suffice to prove that $ \{P \cap Z_i | P \in \mathcal{P}_{Z_i} \} $ admits a finite e.c. decomposition in $ (Z_i, \le) $, for $ i = 1, 2 $, because then (once again) we will only have to merge the two partitions. Fix $ i $. If $ |\{P \cap Z_i | P \in \mathcal{P}_{Z_i} \}| < m $, we conclude by the inductive hypothesis on $ m $. If $ |\{P \cap Z_i | P \in \mathcal{P}_{Z_i} \}| = m $, we conclude by the inductive hypothesis on $ k $: if $ k_i $ is the number of distinct permutations $ \pi $ of $ \{1, \dots, m \} $ such that $ (P_{\pi (1)}, \dots, P_{\pi (m)}) $ occurs in $ Z $, then $ k_i < k $, because $ (P_{\pi_0 (1)}, \dots, P_{\pi_0(m)})) $ occurs in $ Z $ but not in $ Z_i $. \qed
\end{proof}

\begin{remark}\label{rem:special}
Theorem \ref{thm:fdt} follows by picking $ (Z, \le) = (Pref (\mathcal{L(B)}), \preceq) $ and $ \mathcal{P} = \{I_q | q \in Q \} $.
\end{remark}

We say that an e.c. decomposition of $ \mathcal{P} $ in $ (Z, \le) $ is a \emph{minimum-size} e.c. decomposition if it has minimum cardinality among all e.c. decompositions of $ \mathcal{P} $ in $ (Z, \le) $. We will be interested in minimum-size e.c. decompositions because they ensure  additional properties.

 \begin{remark}\label{rem: minimality}
Let $ (Z, \le) $ be a total order, let $ \mathcal{P} = \{P_{1}, \ldots, P_{m} \} $ be a finite partition of $ Z $, and let $ \mathcal V= \{V_{1}, \ldots , V_{r} \}$ be a minimum-size e.c. decomposition of $ \mathcal{P} $ in $ (Z, \le) $.
Then, for every $1\leq i<r$, we have $\mathcal P_{V_i}\not \subseteq  \mathcal P_{V_{i+1}}$, because otherwise $ \mathcal V'= \{V_{1}, \ldots , V_{i-1}, V_i\cup V_{i+1}, \ldots V_{r} \}$ would be an e.c. decomposition of $ \mathcal{P} $ in $ (Z, \le) $ having smaller cardinality. Similarly, for every $ 1< i\leq r $, it must be $\mathcal P_{V_i}\not \subseteq  \mathcal P_{V_{i-1}}$. In other words, for every $ i = 1, \dots, r $ there exist $ P_i \in \mathcal{P}_{V_i} \setminus \mathcal{P}_{V_{i + 1}}$ and $ P'_i \in \mathcal{P}_{V_i} \setminus \mathcal{P}_{V_{i - 1}} $, where we assume $ V_0 = V_{r + 1} = \emptyset $.

In the special case of Remark \ref{rem:special} we conclude that for every $ i = 1, \dots, r $ there exists a state $ q_i $ that occurs in $ V_i $ but not in $ V_{i + 1} $ and a state $ q'_i $ that occurs in $ V_i $ but not in $ V_{i - 1} $.
\end{remark}

In general, a minimum-size e.c. decomposition is not unique.

\begin{example}
Let $ Z $ be the disjoint union of some infinite sets:
\begin{align*}
    P_1 & = \{\alpha, \alpha_1, \alpha_2, \dots, \}; \\
    P_2 & = \{\beta, \beta_1, \beta_2, \dots, \beta'_1, \beta'_2, \dots, \}; \\
    P_3 & = \{\gamma_1, \gamma_2, \dots \}. \\
\end{align*}
Let $ \le $ be the total order on $ Z $ such that:
\begin{equation*}
\alpha_1 < \beta_1 < \alpha_2 < \beta_2 < \dots < \alpha < \beta < \gamma_1 < \beta'_1 < \gamma_2 < \beta'_2 < \dots
\end{equation*}
Let $ \mathcal{P} = \{P_1, P_2, P_3 \} $. Then, two distinct minimum-size e.c. decompositions of $ \mathcal{P} $ in $ (Z, \le) $ are:
\begin{enumerate}
    \item $ \mathcal{V} = \{\{\alpha_1,  \beta_1,\alpha_2, \beta_2, \dots, \alpha, \beta \}, \{\gamma_1, \beta'_1, \gamma_2, \beta'_2, \dots\}\} $;
    \item $ \mathcal{V'} = \{\{\alpha_1,  \beta_1, \alpha_2, \beta_2, \dots, \alpha \}, \{\beta, \gamma_1, \beta'_1, \gamma_2, \beta'_2, \dots\}\} $.
\end{enumerate}
\end{example}

\begin{example}\label{ex:no_uniqueness}
Let us show that even in the special case of Remark \ref{rem:special} minimum-size e.c. decompositions need not be unique.
Consider the DFA $ \mathcal{B} $ in figure \ref{fig:2minent}. Notice that in every e.c. decomposition of $ \mathcal{B} $ one element is $ \{\epsilon \} $, because $ I_0 = \{\epsilon \} $. Moreover, every e.c. decomposition of $ \mathcal{B} $ must have cardinality at least three, because $ I_1 \prec I_3 $. It is easy to check that:
\begin{equation*}
    \mathcal{V} = \{\{\epsilon \}, \{ac^*\cup bc^*\}, \{[b(c+d)^*\setminus bc^*]\cup f(c+d)^*\cup gd^*\}   \}
\end{equation*}
and:
\begin{equation*}
    \mathcal{V}' = \{\{\epsilon \}, \{ac^*\cup b(c+d)^*\cup [f(c+d)^*\setminus fd^*]\}, \{fd^*\cup gd^*\}\}
\end{equation*}
are two distinct minimum-size e.c. decompositions of $ \mathcal{B} $.

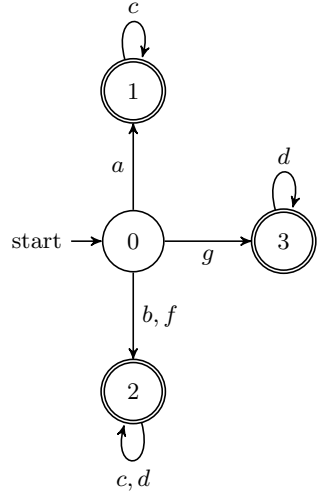
\begin{figure}[h!] 
 \begin{center}
\begin{tikzpicture}[->,>=stealth', semithick, auto, scale=1]
\node[state, initial] (0)    at (0,0)	{$0$};
\node[state, accepting,label=above:{}] (1)    at (0,2)	{$1$};
\node[state,accepting, label=above:{}] (2)    at (0,-2)	{$2$};
\node[state, accepting,label=above:{}] (3)    at (2,0)	{$3$};
\draw (0) edge [] node [] {$a$} (1);
\draw (0) edge [] node [] {$b,f$} (2);
\draw (0) edge [] node [below] {$g$} (3);
 \draw (1) edge  [loop above] node {$c$} (1);
 \draw (2) edge  [loop below] node {$c,d$} (2);
  \draw (3) edge  [loop above] node {$d$} (3);
\end{tikzpicture}
 \end{center}
 	\caption{An automaton $\mt B$ admitting two distinct minimum-size e.c. decompositions.}\label{fig:2minent}
\end{figure}
\end{example}

\subsection{On Some Properties of the Equivalence  $\sim$}\label{app:sim_properties}

\paragraph*{\textbf{Statement of Lemma \ref{lem:sim_finite}.}}
Let $ \mathcal B$ be a DFA and let $\mt V$ be a minimum-size e.c. decomposition of $ Pref (\mathcal{L(B)})$. Then, $\sim_{\mt V}$ has a finite number of classes on  $\pf {L(\mt B)}$ and 
$\mt L(\mt B)$ is equal to the union of some $\sim_\mathcal{V}$-classes.

\paragraph*{\textbf{Proof}} The relation $\sim_{\mt V}$ has finite index because each class is the union of some sets of the type $ V \cap I_q $, with $ V \in \mathcal{V} $ and $ q \in Q $. From the definition of $\sim_\mathcal{V}$ it follows that $\alpha\sim_\mathcal{V}\beta$ implies $\delta(s,\alpha)=\delta(s,\beta)$. Hence,   $\alpha\in \mt {L(\mt B)}$ implies  $\beta\in \mt {L(\mt B)}$, proving  that  $\mt L(\mt B)$ is equal to the union of some $\sim_\mathcal{V}$-classes. \qed

\paragraph*{}
We now prove some general results that will lead us to the proof of Lemma \ref{lem:easy_sim}. The difficult implication is $ (\leftarrow) $. To this end  we will prove that without loss of generality we can assume $ C_1 \prec \dots \prec C_n $. This will follow from some results for convex and entangled sets which hold true in the general setting outlined in Definition \ref{def:entangled}.

\begin{lemma}\label{lem:union}
Let $ (Z, \le) $ be a total order and let $ \mathcal{P} $ be a partition of $ Z $. Let $ C $ be a convex, entangled set. Assume that $ C_1 $ and $ C_2 $ are convex sets such that $C = C_1 \cup C_2 $. Then, there exists $ i \in \{1, 2 \} $ such that $ C_i $ is entangled and $ \mathcal{P}_{ C_i} = \mathcal{P}_{ C} $. 
\end{lemma}
\begin{proof}
Let $ (z_i)_{i \in \mathbb{N}} $ be a monotone sequence witnessing that $ C $ is entangled. Then, at least one of the following is true:
\begin{enumerate}
    \item $ (\forall i \in \mathbb{N})(\exists j \in \mathbb{N})(j > i \land z_j \in C_1) $
    \item $ (\forall i \in \mathbb{N})(\exists j \in \mathbb{N})(j > i \land z_j \in C_2) $.
\end{enumerate}
 In the first case $ C_1 $ is entangled:   by the convexity of $C_1$, if    $ j_0 > 1 $ is such that $ z_{j_0} \in C_1 $ then  the subsequence $ (z_j)_{j \ge j_0} $ is contained in $ C_1 $ an d clearly we have $\mathcal{P}_{ C_1} = \mathcal{P}_{ C}$.  Analogously, in the second case $C_2$ is entangled and  $\mathcal{P}_{ C_2} = \mathcal{P}_{ C}$. \qed
\end{proof}

\begin{lemma}\label{lem:entangled}
Let $ (Z, \le) $ be a total order and let $ \mathcal{P} $ be a partition of $ Z $. Let $ C_1, C_2 $ be convex, entangled sets. Then, at least one the following holds true:
\begin{enumerate}
    \item $ C_1 \setminus C_2 $ is a convex, entangled set and $ \mathcal{P}_{ C_1 \setminus C_2} = \mathcal{P}_{ C_1} $;
    \item $ C_2 \setminus C_1 $ is a convex, entangled set and $ \mathcal{P}_{ C_2 \setminus C_1} = \mathcal{P}_{ C_2} $;
    \item $ C_1 \cup C_2 $ is a convex, entangled set and   $ \mathcal{P}_{ C_1 \cup C_2} = \mathcal{P}_{C_1} = \mathcal{P}_{ C_2}$.
\end{enumerate}
\end{lemma}
\begin{proof}
We can assume $ C_2 \not \subseteq C_1 $, $ C_1 \not \subseteq C_2 $ and $ C_1 \cap C_2 \not = \emptyset $, otherwise the conclusion is immediate (at least one of three sets is equal to $ C_1 $ or $ C_2 $). 
Under these assumptions, it is easy to check that $ C_1 \setminus C_2 $, $ C_2 \setminus C_1 $, $ C_1 \cap C_2 $ and $ C_1 \cup C_2 $ are convex (but note that for example $ C_1 \setminus C_2 $ need not be convex if $ C_2 \subseteq C_1 $). Since $ C_1 = (C_1 \setminus C_2) \cup (C_1 \cap C_2) $ and $ C_2 = (C_2 \setminus C_1) \cup (C_1 \cap C_2) $, by Lemma \ref{lem:union} we conclude that at least one the following holds true:
\begin{enumerate}
    \item $ C_1 \setminus C_2 $ is a convex, entangled set and $ \mathcal{P}_{C_1 \setminus C_2} = \mathcal{P}_{C_1} $;
    \item $ C_2 \setminus C_1 $ is a convex, entangled set and $ \mathcal{P}_{C_2 \setminus C_1} = \mathcal{P}_{C_2} $;
    \item $ C_1 \cap C_2 $ is a convex, entangled set and $ \mathcal{P}_{C_1 \cap C_2} = \mathcal{P}_{C_1} = \mathcal{P}_{C_2}$.
\end{enumerate}
In the third case, we have $ \mathcal{P}_{C_1 \cup C_2} = \mathcal{P}_{C_1} \cup \mathcal{P}_{C_2} = \mathcal{P}_{C_1} = \mathcal{P}_{C_2} = \mathcal{P}_{C_1 \cap C_2} $. Since $C_1\cap C_2 \subseteq C_1\cup C_2$  and $C_1\cap C_2$ is entangled,   we conclude that $ C_1 \cup C_2 $ is entangled. \qed
\end{proof}

We can now state the result that, for $ (Z, \le) = (Pref (\mathcal{L(B)}), \preceq) $ and $ \mathcal{P} = \{I_q | q \in Q \} $, will imply that in Lemma \ref{lem:easy_sim} we can assume $ C_1 \prec \dots \prec C_n $.

\begin{lemma} \label{lem:convex_in_order}
Let $ (Z, \le) $ be a total order and let $ \mathcal{P} $ be a partition of $ Z $. Let $ C_1, \dots, C_n $ be convex,  entangled set. Then, there exist convex, entangled sets $ C'_1, \dots, C'_m $, with $ m \le n $, such that $ C'_1 < \dots < C'_m $  and  $ \cup_{i = 1}^n C_i = \cup_{i = 1}^m C'_i $. Moreover, for every $P\in \mathcal P$ it holds $P\cap C_i\neq \emptyset$ for all $i=1, \ldots, n$ if and only if $P\cap C_i'\neq \emptyset$, for all $i=1, \ldots, m$. 
\end{lemma}
\begin{proof}  We suppose, without loss of generality,  that the $C_i$'s are non-empty.
We proceed by induction on $ r = |\{(i, j)| 1 \le i < j \le n | C_i \cap C_j \not = \emptyset \}| $. If $ r = 0 $, then every two distinct $ C_i $'s are disjoint,  and since they are convex, it follows that  they are comparable;   let $ C'_1, \dots, C'_n $ be the permutation of the $ C_i $'s such that $ C'_1< \dots< C'_n $. Now assume $ r \ge 1 $. Without loss of generality, we can assume that  $ C_1 \cap C_2 \not = \emptyset $. By Lemma \ref{lem:entangled} we know that at least one among $ C_1 \setminus C_2 $, $ C_2 \setminus C_1 $ and $ C_1 \cup C_2 $ is convex and entangled. If $ C_1 \cup C_2 $ is convex and entangled, then consider $ C_1 \cup C_2, C_3, \dots, C_n $ and   notice that the number of intersections between pairs of this sets is smaller (because if $ (C_1 \cup C_2) \cap C_i \not = \emptyset $ for some $ i \ge 3 $, then either $ C_1 \cap C_i \not = \emptyset $ or $ C_2 \cap C_i \not = \emptyset $, so either $ C_1 $  or  $ C_2 $ intersects $C_i$ in the original collection $ C_1, C_2, \dots, C_n $). Now assume that $ C_1 \setminus C_2 $ is entangled (the case $ C_2 \setminus C_1 $ analogous). Consider $ C_1 \setminus C_2, C_2, C_3, \dots, C_n $. Even in this case the number of intersections is smaller, because $ C_1 \setminus C_2 \subseteq C_1 $ and $ (C_1 \setminus C_2) \cap C_2 = \emptyset $.  In both cases Lemma \ref{lem:entangled} implies that a $P \in \mathcal{P}$ occurs in all $C_i$'s if and only if it occurs in all elements of the new family and  we can conclude by the inductive hypothesis. \qed
\end{proof}

We can now tackle the proof of Lemma \ref{lem:easy_sim}.

\paragraph*{\textbf{Statement of Lemma \ref{lem:easy_sim}.}}
 Let $\mt B$ be a DFA and  let $\mt V$ be a minimum-size e.c. decomposition of $\mt B$. 
Then  $\alpha \sim_{\mt V} \alpha'$
    if and only if $\delta(s,\alpha)= \delta(s, \alpha')$ and   the interval $[\alpha, \alpha']^\pm$ is contained in a  finite union of convex, entangled sets  $C_1, \ldots,C_n$ in $\pf {L( \mt B)}$, with $C_i\cap I_{\delta(s,\alpha)}\neq \emptyset$ for all $i=1, \ldots,n$. In particular, $ \sim_{\mt V} $ is independent of the choice of $ \mathcal{V} $.

\paragraph*{\textbf{Proof}} $ (\rightarrow) $ follows by considering all sets $ V \in \mathcal{V} $ such that $ V \cap [\alpha, \alpha']^\pm \not = \emptyset $.

Let us prove $ (\leftarrow)$. Pick $ \alpha, \alpha' \in Pref (\mathcal{L(B)}) $ such that $\delta(s,\alpha)= \delta(s, \alpha') $ and the interval $[\alpha, \alpha']^\pm$ is contained in a  finite union of convex, entangled sets $C_1, \ldots,C_n$, with $C_i\cap I_q\neq \emptyset$ for every $i=1, \ldots,n$, where $ q = \delta(s,\alpha)= \delta(s, \alpha') $. By Lemma \ref{lem:convex_in_order} we can assume that $C_1\prec C_2 \prec \ldots \prec C_n$. Assume without loss of generality that $ \alpha \prec \alpha' $, and consider $ V \in \mathcal{V} $ such that $ \alpha \prec V \prec \alpha' $. We must prove that $ V\cap I_q\neq \emptyset $. It will suffice to prove that if $V\in \mt V$ is such that $V\subseteq   \bigcup_{i=1}^n C_i$, then $ q $ occurs in $ V $.

Let $ \mathcal{V} = \{V_1, \dots, V_r \} $, with $ V_1 \prec \dots \prec V_r $. Since $ \mathcal{V} $ is minimum-size, by Remark \ref{rem: minimality} for every $ i = 1, \dots, r $ there exist two states $ q_i $ and $ q'_i $ such that $ q_i $ occurs in $ V_i $ but not in $ V_{i + 1} $ and $ q'_i $ occurs in $ V_i $ but not in $ V_{i - 1} $ (we assume $ V_0 = V_{r + 1} = \emptyset $).

Assume that $ V = V_s $. If $ C_i \subseteq V_s $ for some $ i $, then the conclusion follows because $C_i\cap I_q\neq \emptyset$. Otherwise, since $C_1\prec C_2 \prec \ldots \prec C_n$, it must be $ V_s \subseteq C_i \cup C_{i + 1} $ for some $ i $. We distinguish three cases.
\begin{enumerate}
    \item $ V_s \cap  C_{i+1}=\emptyset $. In this case, it must be $ V_s \subseteq C_i $. Let $ V_{s - h}, V_{s - h +1},  \ldots, V_s, \ldots, V_{s+k -1}, V_{s+k} $ ($ h, k \ge 0 $) be all elements of $ \mathcal{V} $ contained in $ C_i $. Since $ V_1 \prec \dots \prec V_r $, we conclude:
    \begin{equation*}
        V_{s - h }\cup \ldots  \cup V_{s}\cup \ldots \cup V_{s+k}\subseteq C_i  \subseteq V_{s-h-1 }\cup V_{s-h}\cup \ldots\cup V_{s}\cup \ldots \cup V_{s+k}\cup V_{s+k+1 }
    \end{equation*}
    where as usual $ V_0 = V_{r + 1} = \emptyset $. We know that $q'_{s - h}, \ldots,  q'_{s},q_{s}, \ldots q_{s+k}$ occur in $V_{s-h}\cup \ldots \cup V_{s+k}$, so they also occur in $C_i$. Moreover, we also know that $ q $ occurs in $ C_i $. Since $ C_i $ is entangled, in particular there exists a sequence $ \alpha_{s - h} \preceq \dots \preceq \alpha_{s} \preceq \beta \preceq \gamma_{s} \preceq \dots \preceq \gamma_{s + k} $ witnessing that the sequence of states $ q'_{s - h}, \ldots,  q'_{s}, q, q_{s}, \ldots q_{s+k} $ occurs in $ C_i $, and so also in $ V_{s-h -1 }\cup V_{s-h}\cup \ldots\cup V_{s}\cup \ldots \cup V_{s+k}\cup V_{s+k+1 } $. Since $ q'_{s - h} $ does not occur in $ V_{s - h - 1} $, then the sequence $ q'_{s - h + 1}, \ldots,  q'_{s}, q, q_{s}, \ldots q_{s+k} $ occurs in this order in $ V_{s - h}\cup \ldots\cup V_{s}\cup \ldots \cup V_{s+k}\cup V_{s+k+1 } $. Now, $ q'_{s - h + 1} $ does not occur in $ V_{s - h} $, so the sequence $ q'_{s - h + 2}, \ldots,  q'_{s}, q, q_{s}, \ldots q_{s+k} $ occurs in this order in $ V_{s - h + 1}\cup \ldots\cup V_{s}\cup \ldots \cup V_{s+k}\cup V_{s+k+1 } $. Proceeding like that, we obtain that the sequence $ q, q_{s}, \ldots q_{s+k} $ occurs in this order in $ V_{s} \cup \ldots\cup V_{s+k}\cup V_{s+k+1} $. Now suppose by contradiction that $ q $ does not occur in $ V_s $. Then, as before we obtain that $ q_{s}, \ldots q_{s+k} $ occurs in this order in $ V_{s + 1}\cup \ldots\cup V_{s+k}\cup V_{s+k+1 } $, then $ q_{s+ 1}, \ldots q_{s+k} $ occurs in this order in $ V_{s + 2}\cup \ldots\cup V_{s+k}\cup V_{s+k+1 } $, and we finally conclude that $ q_{s + k} $ occurs in $ V_{s+k+1 } $, a contradiction.
    \item $ V_s \cap  C_{i}=\emptyset $. In this case, it must be $ V \subseteq C_{i + 1} $ and one concludes as in the previous case.
    \item $ V_s \cap  C_{i} \not = \emptyset $ and $ V\cap  C_{i + 1} \not = \emptyset $. In this case, let $ V_{s - h}, \dots, V_{s - 1} $ ($ h \ge 0 $) be all elements of $ \mathcal{V} $ contained in $ C_i $, and let $ V_{s + 1}, \dots, V_{s + k} $ ($ k \ge 0 $) be all elements of $ \mathcal{V} $ contained in $ C_{i + 1} $. As before:
    \begin{equation*}
    V_{s - h} \cup \dots \cup V_{s - 1} \subseteq C_i \subseteq V_{s -h - 1} \cup V_{s - h} \cup \dots \cup V_{s - 1} \cup V_s
    \end{equation*}
    and:
    \begin{equation*}
    V_{s + 1} \cup \dots \cup V_{s + k} \subseteq C_i \subseteq V_s \cup V_{s + 1} \cup \dots \cup V_{s + k} \cup V_{s + k + 1}.
    \end{equation*}
    Now, assume by contradiction that $ q $ does not occur in $ V_s $. First, let us prove that $ q'_s $ does not occur in $ C_i $. Suppose by contradiction that $ q'_s $ occurs in $ C_i $. We know that $ q'_{s - h}, \dots, q'_{s - 1} $ occurs in $ C_i $, and we also know that $ q $ occurs in $ C_i $. Since $ C_i $ is entangled, then $ q'_{s - h}, \dots, q'_{s - 1}, q'_s, q $ should occur in this order in $ C_i $ and so also in $ V_{s -h - 1} \cup V_{s - h} \cup \dots \cup V_{s - 1} \cup V_s $, which as in case 1 would imply that $ q $ occurs in $ V_s $, a contradiction. Analogously, one shows that $ q_s $ does not occur in $ C_{i + 1} $.
    
    Since $ q_s $ and $ q'_s $ occur in $ V_s $, then there exists a monotone sequence in $ V_s $ whose trace consists of alternating values of $ q_s $ and $ q'_s $. But $ V_s \subseteq C_i \cup C_{i + 1} $ and $ C_i \prec C_{i + 1} $, so the monotone sequence is definitely contained in $ C_i $ or $ C_{i + 1} $. In the first case we would obtain that $ q'_s $ occurs in $ C_i $, and in the second case we would obtain that $ q_s $ occurs in $ C_{i + 1} $, so in both cases we reach a contradiction. \qed
\end{enumerate}

The $ \mathcal{V} $-free characterization of $ \sim_{\mt V} $ allows us to easily deduce right-invariance.

\paragraph*{\textbf{Statement of Lemma \ref{lem:simequivalence}.}} Let $\mt B$ be a DFA. Then, the equivalence relation $\sim$ is right-invariant. 

\paragraph*{\textbf{Proof}} Assume that $\alpha \sim \alpha' $ and $a\in \Sigma$ is such that $\alpha a \in \pf{L }$. We must prove that $\alpha' a\in  \pf{L }$ and  $\alpha a\sim \alpha' a$. By Lemma \ref{lem:easy_sim}, we know that $\delta(s,\alpha)= \delta(s, \alpha')$ and there exist convex, entangled sets $ C_1, \dots, C_n $ such that $[\alpha, \alpha']^\pm \subseteq C_1 \cup \dots \cup C_n $ and $ C_i\cap I_{\delta(s,\alpha)}\neq \emptyset$ for all $i=1, \ldots,n$. We must prove that $\alpha' a\in  \pf{L }$, $\delta(s,\alpha a)= \delta(s, \alpha' a)$ and there exist convex, entangled sets $ C'_1, \dots, C'_{n'} $ such that $[\alpha a, \alpha' a]^\pm \subseteq C'_1 \cup \dots \cup C'_{n'} $ and $ C'_i\cap I_{\delta(s,\alpha a)}\neq \emptyset$ for all $i=1, \ldots,n'$.

From $\delta(s,\alpha)= \delta(s, \alpha')$ and $\alpha a \in \pf{L }$ we immediately obtain $\alpha' a \in \pf{L }$ and $\delta(s,\alpha a)= \delta(s, \alpha' a)$. Moreover, from $[\alpha, \alpha']^\pm \subseteq C_1 \cup \dots \cup C_n $ we obtain $[\alpha a, \alpha' a]^\pm \subseteq C_1a \cup \dots \cup C_na $, and from $ C_i\cap I_{\delta(s,\alpha)}\neq \emptyset$ we obtain $ C_i a \cap I_{\delta(s,\alpha a)}\neq \emptyset $, so we only have to prove that every $ C_i a $ is convex and entangled. As for convexity, let $ \alpha, \beta, \gamma \in Pref (\mathcal{L(B)}) $ such that $ \alpha \prec \beta \prec \gamma $ and $ \alpha, \gamma \in C_i a $. We must prove that $ \beta \in C_i a $. Since $ \alpha, \gamma \in C_i a $, we can write $ \alpha = \alpha' a $ and $ \gamma = \gamma' a $, with $ \alpha', \gamma' \in C_i$. From $ \alpha' a \prec \beta \prec \gamma' a $ we obtain $ \beta = \beta' a $ for some $ \beta \in Pref(\mathcal{L(B)}) $. Since $ \alpha' \prec \beta' \prec \gamma' $ and $ \alpha', \gamma' \in C_i $ then $ \beta' \in C_i $ by convexity and so $ \beta \in C_i a $. Finally, if $ (\alpha_i)_{i \in \mathbb{N}} $ is a monotone sequence witnessing that $ C_i $ is entangled, then $ (\alpha_i a)_{i \in \mathbb{N}} $ is a monotone sequence witnessing that $ C_i a $ is entangled. \qed

\subsection{The Hasse automaton}

Let us start with a fairly intuitive lemma which will be used in the proof of Theorem \ref{thm:equivalentholeproof}.

\begin{lemma}\label{lem:con-int}
Let $ (Z, \le) $ be a total order. If $ C_1, \ldots , C_{n} $ are convex sets such that $ C_{i}\cap C_{j} \neq \emptyset $ for all pairwise distinct $i, j\in \{ 1 , \ldots , n\} $, then $ \bigcap_{i=1}^{n} C_{i}\neq \emptyset $.     
\end{lemma}
\begin{proof}
We proceed by induction on $ n $. Cases $ n=1, 2 $ are trivial, so assume $ n \ge 3 $. For every $ i \in \{1, \ldots , n\} $, the set:
\[\bigcap_{\substack{k \in \{1, \ldots , n\} \\ k \neq i}} C_{k}
\]
is nonempty by the inductive hypothesis, so we can pick an element $ d_i $. If for some distinct $ i $ and $ j $ we have $ d_{i}=d_{j} $, then such an element witnesses that $ \bigcap_{i=1}^{n} C_{i}\neq \emptyset $. Otherwise, assume without loss of generality that $ d_{1} < \dots < d_{n} $. Fix any integer $ j $ such that  $ 1 < j < n $, and let us prove that $ d_j $ witnesses that $ \bigcap_{i=1}^{n} C_{i}\neq \emptyset $. We only have to prove that $ d_j \in C_j $. This follows from $ d_1, d_n \in C_j $ and the convexity of $ C_j $. \qed
\end{proof}

We can now prove Theorem \ref{thm:equivalentholeproof}.

\paragraph*{\textbf{Statement of Theorem \ref{thm:equivalentholeproof}.}}
Let  $\mt B $ be  a DFA. Then, there exists a DFA $\mt B'$ such that $ \mathcal{L(B')} = \mathcal{L(B)} $ and $\text{ent}(\mt B)=\text{ent}(\mt B')= \text{width}(\mt B')$.

\paragraph*{\textbf{Proof}} Let $ \sim $ be the equivalence relation $ \sim_\mathcal{B} $ considered in this section, and define $ \mathcal B' =(Q', s', \delta', F') $ by:
\begin{itemize}
	\item $ Q' = \{  [\alpha]_{\sim}: \alpha \in \pf L\}$;
	\item $ \delta'([\alpha]_{\sim},a) = [\alpha a]_{\sim} $ for every $ \alpha \in Pref (\mathcal{L(B)}) $ and for every $ a \in \Sigma $ such that $ \alpha a \in Pref (\mathcal{L(B)}) $;
	\item $ s' = [ \epsilon ]_\sim $, where $ \epsilon $ is the empty string;
	\item $ F'= \{ [\alpha]_\sim:\alpha \in \mathcal L\} $.
\end{itemize}
Since $ \sim $ is right-invariant, it has finite index and $ \mathcal{L} $ is the union of some $ \sim $-classes (Lemma \ref{lem:simequivalence} and Lemma \ref{lem:sim_finite}), then $ \mathcal{B'} $ is a well-defined DFA. Moreover, it is easy to check that it holds:
\begin{equation}\label{eq30}
\alpha \in [\beta]_\sim \iff \delta' (s', \alpha) = [\beta]_\sim
\end{equation}
which implies that for every $ \alpha \in Pref(\mathcal{L}) $ it holds:
\begin{equation}\label{eq31}
I_{[\alpha]_\sim} = [\alpha]_\sim   
\end{equation}
and so $ \mathcal{L}(\mathcal{B}') = \mathcal{L(\mathcal{B}}) $.

Let us prove that $\text{ent}(\mt B')\leq \text{ent}(\mt B)$. Let $ \{[\alpha_1]_\sim, \dots, [\alpha_{\text{ent}(\mt B')}]_\sim \} $ be an entangled set of (pairwise distinct) states in $ \mathcal{B'} $ having maximum cardinality, as witnessed by some monotone sequence $ (\beta_i)_{i \in \mathbb{N}} $. Let $ \mathcal{V} $ be a minimum-size e.c. decomposition of $\pf {L(B)}$. Since all elements of $ \mathcal{V} $ are convex, then there exists $ V \in \mathcal{V} $ such that $ (\beta_i)_{i \in \mathbb{N}} $ is definitely contained in $ V $, so in particular there exist integers $ i_1, \dots, i_{\text{ent}(\mt B')} $ such that $ \beta_{i_k} \in V $ and $ \delta'(s, \beta_{i_k}) = [\alpha_k]_\sim $ (or equivalently $ [\alpha_k]_\sim = [\beta_{i_k}]_\sim $ by equation \ref{eq30}), for every $ k = 1, \dots, \text{ent}(\mt B') $. Now, define $ q_i = \delta (s, \beta_{i_k}) $. Notice that $ q_1, \dots, q_{\text{ent}(\mt B')} $ are pairwise distinct: if for some distinct integers $ r, s $ it were $ q_r = q_s $, then we would conclude $ \beta_{i_r} \sim \beta_{i_s} $ (because $ \delta(s, \beta_{i_r}) = q_r = q_s = \delta(s, \beta_{i_s}) $ and there is no element of $ \mathcal{V} $ between $\beta_{i_r} $ and $\beta_{i_s} $, being $\beta_{i_r}, \beta_{i_r} \in V $), or equivalently, $ [\alpha_r]_\sim = [\alpha_s]_\sim $, a contradiction because $ \{[\alpha_1]_\sim, \dots, [\alpha_{\text{ent}(\mt B')}]_\sim \} $ consists of pairwise distinct states. Moreover, $ \{q_1, \dots, q_{\text{ent}(\mt B')}\} $ is an entangled set in $ \mathcal{B} $, because all these states occur  in $ V $ (as witnessed by $ \beta_{i_1}, \dots, \beta_{i_k} $) and $ V $ is an element of an e.c. decomposition. In particular, $ \{q_1, \dots, q_{\text{ent}(\mt B')}\} $ witnesses that $\text{ent}(\mt B')\leq \text{ent}(\mt B)$.

Let us prove that $\text{ent}(\mt B)\leq \text{ent}(\mt B')$. Let $ \{q_1, \dots, q_{\text{ent}(\mt B)} \} $ be an entangled set of (pairwise distinct) states in $ \mathcal{B} $ having maximum cardinality, as witnessed by some monotone sequence $ (\alpha_i)_{i \in \mathbb{N}} $. Notice that every $ I_{q_k} $ is equal to a (finite) union of some $ I_{q'} $, with $ q' \in Q' $: if $ \beta\in I_{q_k} $ and $\beta, \beta'\in I_{q'}$, where $q'=[\alpha]_\sim$, then by equation \ref{eq31} it holds $ \beta \sim \beta' $, from which  $ q_k=\delta(s, \beta) = \delta(s, \beta') $ and  $ \beta' \in I_{q_k} $ follow. Since $ (\alpha_i)_{i \in \mathbb{N}} $ goes through $ q_k $ infinitely many times, then there exist $ q'_k \in Q' $ such that $ I_{q'_k} \subseteq I_{q_k} $ and $ (\alpha_i)_{i \in \mathbb{N}} $ goes through $ q'_k $ infinitely many times. We conclude that $ q'_1, \dots, q'_k $ are pairwise distinct and $ \{q'_1, \dots, q'_k \} $ is an entangled set of states in $ \mathcal{B'} $, which implies $\text{ent}(\mt B)\leq \text{ent}(\mt B')$.

    Let us prove that $ \text{width}(\mt B')\leq \text{ent}(\mt B)$. By Dilworth's theorem, there exist states: \[q_1'=[\alpha_1]_\sim, \dots, q_{\text{width}(\mt B')}'=[\alpha_{\text{width}(\mt B')}]_\sim  \] being pairwise not $ \le_\mathcal{B'}$-comparable. Define $ q_i = \delta(s, \alpha_i) $. From Definition  \ref{def:sim} it follows that for all $i=1,\ldots, k$   there exist two integers $ n_i $, $ m_i $, with $n_i\leq m_i$, such that 
\[I_{q_i'}\subseteq V_{n_i}\cup V_{n_i+1} \cup \ldots \cup V_{m_i} \text{ and for every $ n_i \le j \le m_i $ it holds $ V_j \cap I_{q_i}\neq \emptyset $}.\]
Let $W_i=V_{n_i}\cup V_{n_i+1} \cup \ldots \cup V_{m_i}$.  Since, for all $i$,   $W_i$ is a  union of consecutive elements of $\mt V$  and   $I_{q_i}$  is a  subset  of   $W_i$,  
 the incomparability of $q_i', q_j'$ implies $W_i\cap W_j\neq \emptyset$, for all $i\neq j$.  Since all the $ W_{i} $'s are convex, by Lemma \ref{lem:con-int} it follows that $ \bigcap_{i=1}^{m}W_{i} \neq \emptyset $ and, since the $W_i$'s are unions of consecutive elements of the partition $\mathcal V$ containing $q_i$, it follows  that    there exists an element $V$  in $ \mathcal V $ which is a subset of all $ W_{i} $'s. 
But then $V\cap I_{q_i}\neq \emptyset$, for all $i=1, \ldots, \text{width}(\mt B')$, and since $V$ is entangled it follows  that $q_1, \ldots, q_{\text{width}(\mt B')}$ are entangled in $\mt B$.

Summarizing, we proved that    $ \text{width}(\mt B')\leq \text{ent}(\mt B)= \text{ent}(\mt B')$, so that by   Lemma \ref{lem:ent_leq_width}  we obtain  $\text{ent}(\mt B)=\text{ent}(\mt B')= \text{width}(\mt B')$. \qed

\paragraph*{}

Here is the proof of our main theorem. 

\paragraph*{\textbf{Statement of Theorem \ref{thm:hasse}.}}
If $\mt A$ is the minimum automaton of the regular language $\mt L$, then:
\[\text{width}(\mt L)= \text{width}(\mt H_{\mt L})= \text{ent}(\mt A)=\text{ent}(\mt L).\]

\paragraph*{\textbf{Proof}}
   By  Lemma \ref{lem:minent} we have     
   $\text{ent}(\mt A)= \text{ent}(\mt L)$. Since $\text{ent}(\mt B)\leq \text{width}(\mt B)$ for all DFAs (Lemma \ref{lem:ent_leq_width}), we obtain $\text{ent}(\mt L)\leq \text{width}(\mt L), $ and from   Theorem \ref{thm:equivalentholeproof} it follows $ \text{width}(\mt H_{\mt L})= \text{ent}(\mt A) $. We have:
     \[\text{width}(\mt H_{\mt L})= \text{ent}(\mt A)= \text{ent}(\mt L) \leq \text{width}(\mt L)\leq \text{width}(\mt H_{\mt L}) \]
    and the conclusion follows. \qed

\paragraph*{\textbf{Statement of Lemma \ref{lem:minHasse}.}}  Let $ \mathcal{L} $ be a language, and consider the class:
 \begin{equation*}
\mathscr{C} = \{\mathcal{B'} | \text{$ \mathcal{B} $ is a DFA and $ \mathcal{L(B)} = \mathcal{L} $} \}.
 \end{equation*}
 Then, there exists exactly one DFA being in $ \mathscr{C}$ and having the minimum number of states, namely, the Hasse automaton  $\mt H_{\mt L}$. In other words, $\mt H_{\mt L}$ is the minimum DFA of $ \mathscr{C} $.

\paragraph*{\textbf{Proof}}
 Let $ \mt A $ the minimum DFA recognizing $ \mathcal{L} $ and let $ \mathcal{B} $ any DFA recognizing $ \mathcal{L} $. Let us prove that $ \sim_\mathcal{B} $ is a refinement of $ \sim_\mathcal{A} $. If $ \alpha \sim_\mathcal{B} \alpha' $, then by Lemma \ref{lem:easy_sim} we have $\delta(s_\mt B,\alpha)= \delta(s_\mt B, \alpha')$ and   the interval $[\alpha, \alpha']^\pm$ is contained in a  finite union of convex, entangled (with respect to $ \mathcal{B} $) sets  $C_1, \ldots,C_n$ in $\pf {L}$, with $C_i\cap I_{\delta(s_\mt B,\alpha)}\neq \emptyset$  for all $i=1, \ldots,n$. Since $ \mathcal{A} $ is the minimum DFA (and so for every $ q \in Q_\mathcal{A} $ we have that $ I_q $ is the union of some sets $ I_{q'} $, with $ q' \in Q_\mathcal{B} $), then we also have $\delta(s_\mt A,\alpha)= \delta(s_\mt A, \alpha') $, every $ C_i $ is also entangled with respect to $ \mathcal{A} $ and $ C_i\cap I_{\delta(s_\mt A,\alpha)}\neq \emptyset$, so again by Lemma \ref{lem:easy_sim} we conclude $ \alpha \sim_\mathcal{A} \alpha'$. By definition the number of classes of $ \mathcal{B'} $ is equal to index of $ \sim_\mathcal{B} $ and the number of classes of $ \mt H_{\mt L} $ is equal to index of $ \sim_\mathcal{A} $, so $ \mt H_{\mt L}$ is a minimum DFA of $ \mathscr{C} $, being $ \mathcal{B} $ arbitrary. Conversely, if $ \mathcal{B}' $ is a minimum DFA of $ \mathscr{C} $, then $ \sim_\mathcal{A} $ and $ \sim_\mathcal{B} $ are the same equivalence relation and so by construction $ \mathcal{B'} $ and  $ \mt H_{\mt L}$ are the same DFA. \qed

\section{Proofs of Section \ref{sec:computing}}\label{app:computing}

Let us start with a lemma which will be used to provide an upper bound to our dynamic programming algorithm.
  
\begin{lemma} \label{lem:bound}
Let $ \mathcal{A} $ be an NFA  with set of states $ Q $, and let $ q_1, \dots, q_h \in Q $. If there exist $ \nu_1, \dotsm \nu_h \in Pref (\mathcal{L(A)}) $ such that:
\begin{enumerate}
    \item $ \nu_i \in I_{q_i} $ for every $ i = 1, \dots, h $;
    \item $ \nu_1 \prec \dots \prec\nu_h $;
\end{enumerate}
then, there exist $ \nu'_1, \dotsm \nu'_h \in Pref (\mathcal{L(A)}) $ such that:
\begin{enumerate}
    \item $ \nu'_i \in I_{q_i} $ for every $ i = 1, \dots, h $; 
    \item $ \nu'_1 \prec \dots \prec \nu'_h $;
    \item $ |v'_i | \le h - 2 + \sum_{t = 1}^h |Q|^t $ for every $ i = 1, \dots, h $.  
\end{enumerate}
\end{lemma}

\begin{proof}

For notational simplicity, we will prove the theorem for $ h = 3 $ (the extension to the general case is straightforward).

Given $ \varphi \in \Sigma^* $, we denote by $ \varphi(k)$ the $k$-th letter of $ \varphi $ from the right (if $ |\varphi| < k $ we write $ \varphi (k) = \epsilon $, where $ \epsilon $ is the empty string). For example, $\varphi(1)$ is the last letter of $\varphi$). 

Let $\nu_1\prec\nu_2\prec \nu_3$ be strings in $ I_{q_1}, I_{q_2}, I_{q_3} $, respectively. Let  $d_{3,2}$ be the first position from the right where $\nu_3$ and $\nu_2$ differ.  Since $\nu_3\succ \nu_2$, we have $|\nu_3|\geq d_{3,2}$.
Similarly,  let $d_{2,1}$ be the first position from the right in which $\nu_2$ and $\nu_1$ differ. Again, since $\nu_2 \succ \nu_1$, we have $|\nu_2|\geq  d_{2,1}$.
We distinguish three cases.
\begin{enumerate}
\item $ d_{3, 2} = d_{2, 1} $.
\item $d_{3,2} < d_{2,1}$ (see Figure \ref{figura-i}).
\item $ d_{2,1}<d_{3,2}$ (see Figure \ref{figura-ii}).
  \end{enumerate}
We will not prove case 3 because it is analogous to case 2 (just consider Figure \ref{figura-ii} rather than Figure \ref{figura-i}). Moreover, as we will see, case 1 is not the bottleneck of the length bound. Hence, in the following we assume $d_{3,2} \le d_{2,1}$. Since $|\nu_3|\geq  d_{3,2}$ and $|\nu_2| \geq  d_{2,1}\geq d_{3,2}$, then $\nu_3, \nu_2$, and $ \nu_1$ end with  the same word $\xi$ with $|\xi|= d_{3,2}-1$. Note that $\nu_2\prec \nu_3$ implies $ \nu_2(d_{3,2})\prec  \nu_3(d_{3,2})$. To sum up, we can write:
\begin{equation*}
\nu_1=\theta_1 \nu_1(d_{3,2}) \xi\prec \nu_2=\theta_2  \nu_2(d_{3,2}) \xi \prec  \nu_3=\theta_3 \nu_3(d_{3,2})\xi
\end{equation*}
for some strings $ \theta_1, \theta_2, \theta_3 $.

\begin{figure}[H]
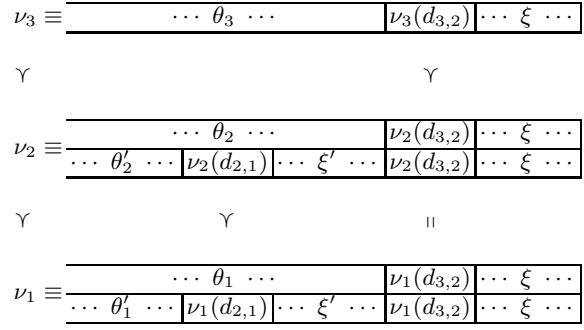

\centering
\begin{tabular}{lc|c|c|c|c|}
\cline{2-6}
$ \nu_3 \equiv $ & \multicolumn{3}{c|}{$ \cdots \  \theta_3 \ \cdots $}  & $ \nu_3(d_{3,2}) $ & $ \cdots \  \xi \  \cdots $ \\ \cline{2-6}
\multicolumn{6}{c}{}\\ 
$ \curlyvee $ &\multicolumn{3}{c}{ }& \multicolumn{1}{c}{$ \curlyvee $} & \multicolumn{1}{c}{}\\ 
\multicolumn{6}{c}{}\\ \cline{2-6}
\multirow{2}{*}{$ \nu_2 \equiv $} & \multicolumn{3}{c|}{$\cdots \  \theta_2 \ \cdots  $} & $ \nu_2(d_{3,2}) $ & $ \cdots \  \xi \  \cdots $ \\ \cline{2-6}
 & $ \cdots \  \theta_{2}'\ \cdots $ & $ \nu_2(d_{2,1}) $ & $ \cdots \ \xi' \  \cdots $ & $ \nu_2(d_{3,2}) $ & $ \cdots \  \xi \  \cdots $ \\ \cline{2-6}
\multicolumn{6}{c}{}\\ 
$ \curlyvee $ &\multicolumn{3}{c}{$ \curlyvee $ }& \multicolumn{1}{c}{$ \shortparallel$} & \multicolumn{1}{c}{}\\ 
\multicolumn{6}{c}{}\\ \cline{2-6}
\multirow{2}{*}{$ \nu_1 \equiv $} & \multicolumn{3}{c|}{$\cdots \  \theta_1 \ \cdots  $} & $ \nu_1(d_{3,2}) $ & $ \cdots \  \xi \  \cdots $ \\ \cline{2-6}
& $ \cdots \  \theta'_{1}\ \cdots $ & $ \nu_1(d_{2,1}) $ & $ \cdots \ \xi' \  \cdots $ & $ \nu_1(d_{3,2}) $ & $ \cdots \  \xi \  \cdots $ \\ \cline{2-6}
\end{tabular}
\caption{Case $ d_{3,2} < d_{2,1} $.  }\label{figura-i}
\end{figure}
 
Without loss of generality, we may assume that   $|\xi|  \leq |\mathcal Q|^3$. Indeed, 
 if $|\xi|>|\mathcal Q|^3$  then when we consider the triples of states visited while reading the last  $| \xi |$ letters in a computation  of  $\nu_1, \nu_2, \nu_3$ we must meet  a repetition, so we could erase a common factor from $\xi$, obtaining a shorter word $\xi_1$ such that 
 $\theta_1 \nu_1(d_{3,2})\xi_1\prec \theta_2 \nu_2(d_{3,2}) \xi_1\prec \theta_3\nu_3(d_{3,2}) \xi_1$, with the three strings still ending in $q_1,q_2$, and $q_3$, respectively.
 
Consider first the case  $d_{3,2}=d_{2,1}$. Let $ s_1, s_2, s_3 $ be the states reached from $ s $ by reading $ \theta_1, \theta_2, \theta_3 $, respectively. Since   $d_{2,1}=d_{3,2} $, then we have $ \nu_1(d_{3, 2}) = \nu_1(d_{2, 1}) \prec \nu_2(d_{2, 1}) = \nu_2 (d_{3, 2}) $ so without loss of generality we may assume that $ \theta_1, \theta_2, \theta_3$   label simple paths from $ s $ to $ s_1, s_2, s_3 $, respectively, and we still have $ \nu_1 \prec \nu_2 \prec \nu_3 $. In other words, we can assume $ | \theta_1|, |\theta_2|, |\theta_3|\leq |\mathcal Q| - 1 $. Hence, in this case we can find $ \nu'_1, \nu'_2, \nu'_3 $ that satisfy the same properties of $ \nu_1, \nu_2, \nu_3 $ and moreover $|\nu_1|, |\nu_2|, |\nu_3|\leq |\mathcal Q| +|\mathcal Q|^3$. 
 
Now consider the case  $d_{3,2}< d_{2,1}$. In this case we have $ \nu_1(d_{3, 2}) = \nu_2 (d_{3, 2}) $. Moreover, $\theta_1 $ and $ \theta_2$ end  with  the same word $\xi'$   with $|\xi'|= d_{2,1}-d_{3,2}-1$, and we can write
\[ \theta_1=\theta_1' \nu_1(d_{2,1})\xi' \prec  \theta_2=\theta'_2\nu_2(d_{2,1})\xi' \]
for some strings $ \theta'_1, \theta'_2 $. Without loss of generality, we can assume that   $|\xi'| \leq |\mathcal Q|^2 $ by arguing as before. Moreover, we have $ \nu_1 (d_{2, 1}) \prec \nu_2 (d_{2, 1})$, so again as before we can assume that $ \theta'_1, \theta'_2 $ and $ \theta_3 $ are simple paths, and so $|\theta'_{1}|,  |\theta'_{2}|, |\theta_{3}|\leq |\mathcal Q| - 1 $.  Hence, in this case we can find $ \nu'_1, \nu'_2, \nu'_3 $ that satisfy the same properties of $ \nu_1, \nu_2, \nu_3 $ and moreover $|\nu_1|, |\nu_2|, |\nu_3|\leq 1 + |\mathcal Q| +|\mathcal{Q}|^2 + |\mathcal Q|^3$. \qed

\begin{figure}[h!]
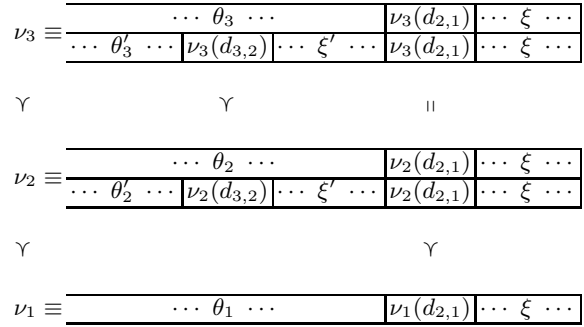

\centering
\begin{tabular}{lc|c|c|c|c|}
\cline{2-6}
\multirow{2}{*}{$ \nu_3 \equiv $} & \multicolumn{3}{c|}{$\cdots \  \theta_3 \ \cdots  $} & $ \nu_3(d_{2,1}) $ & $ \cdots \  \xi \  \cdots $ \\ \cline{2-6}
& $ \cdots \  \theta'_3\ \cdots $ & $  \nu_3(d_{3,2}) $ & $ \cdots \ \xi' \  \cdots $ & $ \nu_3(d_{2,1}) $ & $ \cdots \  \xi \  \cdots $ \\ \cline{2-6}
\multicolumn{6}{c}{}\\ 
$ \curlyvee $ &\multicolumn{3}{c}{$ \curlyvee $ }& \multicolumn{1}{c}{$ \shortparallel $} & \multicolumn{1}{c}{}\\  
\multicolumn{6}{c}{}\\ \cline{2-6}
\multirow{2}{*}{$ \nu_2 \equiv $} & \multicolumn{3}{c|}{$\cdots \  \theta_2 \ \cdots  $} & $ \nu_2(d_{2,1}) $ & $ \cdots \  \xi \  \cdots $ \\ \cline{2-6}
 & $ \cdots \  \theta'_2\ \cdots $ & $ \nu_2(d_{3,2}) $ & $ \cdots \ \xi' \  \cdots $ & $ \nu_2(d_{2,1}) $ & $ \cdots \  \xi \  \cdots $ \\ \cline{2-6}
\multicolumn{6}{c}{}\\ 
$ \curlyvee $ &\multicolumn{3}{c}{ }& \multicolumn{1}{c}{$ \curlyvee $} & \multicolumn{1}{c}{}\\
\multicolumn{6}{c}{}\\ \cline{2-6}
$ \nu_1 \equiv $ & \multicolumn{3}{c|}{$ \cdots \  \theta_1 \ \cdots $}  & $ \nu_1(d_{2,1}) $ & $ \cdots \  \xi \  \cdots $ \\ \cline{2-6}
\multicolumn{6}{c}{}\\ 
\end{tabular}
\caption{Case $ d_{2,1}<d_{3,2}$. }\label{figura-ii}
\end{figure}

\end{proof}
  
We can now prove our graph-theoretical characterization of the width.

\paragraph*{\textbf{Statement of Theorem \ref{thm:conditions width}.}}
    Let $ \mathcal{L} $ be a regular language, and let $ \mathcal{A} $ the minimum DFA of $ \mathcal{L} $, with set of states $ Q $. Let $ k \ge 2 $ be an integer. Then, $ width(\mathcal{L}) \ge k $ if and only if there exist strings $ \mu_1, \dots, \mu_k $ and $ \gamma $ and there exist pairwise distinct $ u_1, \dots, u_k \in Q $ such that:
    \begin{enumerate}
        \item $ \mu_j $ labels a path from the initial state $ s $ to $ u_j $, for every $ j = 1, \dots, k $;
        \item $ \gamma $ labels a cycle starting (and ending) at $ u_j $, for every $ j = 1, \dots, k $;
        \item either all the $ \mu_j $'s are smaller than $ \gamma $ or $ \gamma $ is smaller than all $ \mu_j $'s;
        \item $ \gamma $ is not a suffix of $ \mu_j $, for every $ j = 1, \dots, k $.
    \end{enumerate}

\paragraph*{\textbf{Proof}} First, note that by Theorem \ref{thm:hasse} we have $ width(\mathcal{L}) = ent (\mathcal{A}) $. Let us prove that  if the stated conditions hold true, then $ ent (\mathcal{A}) \ge k $. Notice that for every integer $ i $ we have $ \mu_j \gamma^i \in I_{u_j } $. Moreover, the $ \mu_j $'s are pairwise distinct because the $ u_j $'s are pairwise distinct, so without loss of generality we can assume $ \mu_1 \prec \dots \prec \mu_k $.
\begin{enumerate}
    \item If $ \mu_1 \prec \dots \prec \mu_k \prec \gamma $, consider the increasing sequence:
    \begin{equation*}
        \mu_1 \prec \dots \prec \mu_k \prec \mu_1 \gamma \prec \dots \mu_k \gamma \prec \mu_1 \gamma^2 \prec \dots \mu_k \gamma^2 \prec \mu_1 \gamma^3 \prec \dots \mu_k \gamma^3 \dots
    \end{equation*}
    \item If $ \gamma \prec \mu_1 \prec \dots \prec \mu_k $, consider the decreasing sequence:
        \begin{equation*}
        \mu_k \succ \dots \succ \mu_1 \succ \mu_k \gamma \succ \dots \mu_1 \gamma \succ \mu_k \gamma^2 \succ \dots \mu_1 \gamma^2 \succ \mu_k \gamma^3 \succ \dots \mu_1 \gamma^3 \dots
    \end{equation*}
    where for example $ \mu_k \gamma \prec \mu_1 $ because $ \gamma \prec \mu_1 $ and $ \gamma $ is not a suffix of $ \mu_1 $.
\end{enumerate}
The sequence witnesses that $ \{u_1, \dots, u_k \} $ is an entangled set of distinct states, so $ ent (\mathcal{A}) \ge k $.

Conversely, assume that $ ent (\mathcal{A}) \ge k $. This means that there exist distinct states $ u_1, \dots, u_k $ and there exists a monotone sequence  $ (\alpha_i)_{i \in \mathbb{N}} $ that goes through each state infinitely many times. Since the alphabet $ \Sigma $ is finite, up to removing a finite number of initial elements, we can assume that all $ \alpha_n $'s end with the same $ m = |Q|^k $ letters (in particular, all $ \alpha_n $'s have at least length $ m $), so we can write $ \alpha_n = \alpha'_n \theta $, for some $ \theta \in \Sigma^m $. Since $ |Q| $ is finite, up to taking a subsequence we can assume that the sequence still goes through each state infinitely many times, and all the $ \alpha_i $'s ending in the same state share the sorted sequence of the last $ m + 1 $ states of their path on $ \mathcal{A} $ (that is, the states through which we read $ \theta $). Finally, up to take a subsequence we can assume that $ \alpha_i \in I_{u_j} $ if and only if $ j - i $ is a multiple of $ k $, that is, $ \alpha_1 $, $ \alpha_{k + 1} $, $ \alpha_{2k + 1} $, $ \dots $ are in $ I_{u_1} $, $ \alpha_2 $, $ \alpha_{k + 2} $, $ \alpha_{2k + 2} $, $ \dots $ are in $ I_{u_2} $, and so on.

Let $ x^j_0, x^j_1, \dots, x^j_m $, with $ u_j = x^j_m $ the last $ m + 1 $ states of the path labeled $ \alpha_j $ starting from the initial state, for $ j = 1, \dots, k $. By the properties of the monotone sequence, these are also the last $ m + 1 $ states of the path labeled $ \alpha_i $, for every $ i $ such that $ j - i $ is a multiple of $ k $.

Notice the for every $ 0 \le s \le m $ states $ x^j_s $'s are pairwise distinct, otherwise some $ u_j $'s would be equal (because all $ \alpha_n $'s end with the same $ m $ letters). Moreover, for every $ 0 \le s \le m $ the tuple $ (x^1_s, \dots, x^k_s) $ consists of $ k $ states, and we have $ m + 1 = |Q|^k + 1 $ such tuples, so two tuples must be equal. In other words, there exist $ h, h' $, with $ 0 \le h < h' \le m $, such that $ x^j_h = x^j_{h'} $, for every $ j =1, \dots, k $. This means that between $ x^j_h $ and $ x^j_{h'} $ we read a cycle for every $ j $, and this cycle is the same for every $ j $ (because $ \theta $ is the same for every $ j $). Call this cycle $ \gamma' $, write $ \theta = \phi \gamma' \psi $ for some strings $ \phi $ and $ \psi $, and define $ \beta_i = \alpha'_i \phi $. The sequence $ (\beta_i)_{i \in \mathbb{N}} $ is monotone as well, so from the first $ 2k $ elements of this sequence we can pick $ 2k - 1 $ elements $ \delta_1, \dots, \delta_{2k - 1} $ such that:
\begin{enumerate}
\item $ \delta_1 \prec \dots \prec \delta_{2k - 1} $;
\item $ \delta_i $ and $ \delta_{k + i} $ end in the same state $ u_i $, for $ i = 1, \dots, k - 1 $.
\item $ u_1, \dots, u_{k - 1} $ and the state $ u_k $ where $ \delta_k $ ends are pairwise distinct.
\item For every $ i = 1, \dots, k $ there is a cycle labeled $ \gamma' $ starting at $ u_i $.
\end{enumerate}
Let $ r $ be an integer such that $ |(\gamma')^r| > |\delta_i| $ for every $ i = 1, \dots, 2k - 1 $. Then $ \gamma = (\gamma')^r $ is again a cycle starting at $ u_i $, for every $ i = 1, \dots, k $, and $ \gamma $ is not a suffix of $ \delta_i $, for every $ i = 1, \dots, 2k - 1 $. We distinguish two cases:
\begin{enumerate}
    \item $ \delta_{k} \prec \gamma $. In this case, let $ \mu_1, \dots, \mu_k $ be equal to $ \delta_1, \dots, \delta_k $.
    \item $ \gamma \prec \delta_{k} $. In this case, let $ \mu_1, \dots, \mu_k $ be equal to $ \delta_k, \dots, \delta_{2k - 1} $.
\end{enumerate}
The conclusion follows with these choices for $ \mu_1, \dots, \mu_k $, $ \gamma $, $ u_1, \dots, u_{k} $. \qed

\paragraph*{}

By using Lemma \ref{lem:bound} we can state a computable variant of Theorem \ref{thm:conditions width}.

\begin{corollary}\label{cor:conditions width}
    Let $ \mathcal{L} $ be a regular language, and let $ \mathcal{A} $ the minimum DFA of $ \mathcal{L} $, with set of states $ Q $. Let $ k \ge 2 $ be an integer. Then, $ width(\mathcal{L}) \ge k $ if and only if there exist strings $ \mu_1, \dots, \mu_k $ and $ \gamma $ and there exist pairwise distinct $ u_1, \dots, u_k \in Q $ such that:
    \begin{enumerate}
        \item $ \mu_j $ labels a path from the initial state $ s $ to $ u_j $, for every $ j = 1, \dots, k $;
        \item $ \gamma $ labels a cycle starting (and ending) at $ u_j $, for every $ j = 1, \dots, k $;
        \item either all the $ \mu_j $'s are smaller than $ \gamma $ or $ \gamma $ is smaller than all $ \mu_j $'s;
        \item $ |\mu_1|, \dots, |\mu_k| < |\gamma| \le 2k - 3 + |Q|^k + \sum_{t = 1}^{2k - 1} |Q|^t $.
    \end{enumerate}
\end{corollary}

\begin{proof}
$ (\leftarrow) $ follows from Theorem \ref{thm:conditions width} because 4) implies that $ \gamma $ is not a suffix of any $ \mu_j $. As for $ (\rightarrow) $, it suffices to push forward the proof of Theorem \ref{thm:conditions width}. Notice that by Lemma \ref{lem:bound} we can assume $ |\delta_i | \le 2k - 3 + \sum_{t = 1}^{2k - 1} |Q|^t $ for every $ i = 1, \dots, 2k - 1 $ (and so $ |\mu_j| \le 2k - 3 + \sum_{t = 1}^{2k - 1} |Q|^t $ for every $ j = 1, \dots, k $). Moreover, notice that $ |\gamma'| \le |Q|^k $ (because $ |\theta| = |Q|^k $), and if we pick the minimum $ r $ such that $ |(\gamma')^r| > |\delta_i| $ for every $ i = 1, \dots, 2k - 1 $, then we obtain $ |\gamma| \le 2k - 3 + |Q|^k + \sum_{t = 1}^{2k - 1} |Q|^t $. \qed
\end{proof}

\paragraph*{\textbf{Statement of Theorem \ref{thm:dyn prog}.}}
 Let $\mathcal L$ be  a regular language, given as input by means of any DFA $\mathcal A = (Q, s, \delta,F)$ recognizing $ \mathcal{L}$. 
Then, $p = \text{width}(\mathcal L)$ is  computable in time $|Q|^{O(p)}$.
\begin{proof}
 We exhibit a dynamic programming algorithm based on Corollary \ref{cor:conditions width}. The exact value for $\text{width}(\mathcal L)$ is then found by exponential and binary searches on $k$,  testing the conditions of  Corollary \ref{cor:conditions width} for $O(\log (\text{width}(\mathcal L)))$ values of $k$.
 
 Up to minimizing $ \mathcal{A} $ we can assume that $ \mathcal{A} $ is the minimum DFA recognizing $ \mathcal{L} $. Let $N' = 2k - 3 + |Q|^k + \sum_{t = 1}^{2k - 1} |Q|^t$ be the upper-bound to the lengths of the strings $\mu_i$ ($1\leq i \leq k$) and $\gamma$ that need to be considered, and let $N = N'+1$ be the number of states in a path labeled by a string of length $N'$. Asymptotically, note that $N \in O(|Q|^{2k})$.
 The high-level idea of the algorithm is as follows. First, in condition (3) of Corollary \ref{cor:conditions width}, we focus on finding paths $\mu_j$'s smaller than $\gamma$, as the other case (all $\mu_j$'s larger than $\gamma$) can be solved with a symmetric strategy. Then:
 \begin{enumerate}
     \item For each state $u$ and each length $2\leq \ell \leq N$, we compute the co-lexicographically smallest path of length (number of states) $\ell$ connecting $s$ with $u$.
     \item For each $k$-tuple $u_1,\dots, u_k$ and each length $\ell \leq N$, we compute the co-lexicographically largest string $\gamma$ labeling $k$ cycles of length (number of states) $\ell$ originating (respectively, ending) from (respectively, in) all the states $u_1,\dots, u_k$.
 \end{enumerate}
 
Steps (1) and (2) could be naively solved by enumerating the strings $\mu_1, \dots, \mu_k$, and $\gamma$ and trying all possible combinations of states $u_1, \dots, u_k$. Because of the string enumeration step, however, this strategy would be exponential in $N$, i.e. doubly-exponential in $k$. We show that a dynamic programming strategy is exponentially faster. 

\texttt{Step (1)}. This construction is identical to the one  used  in \cite{alanko2020wheeler} for the Wheeler case $ (p = 1) $. For completeness, we report it here. 
Let $\pi_{u,\ell}$, with $u\in Q$ and $2 \leq \ell \leq N$, denote the predecessor of $u$ such that the co-lexicographically smallest path of length (number of states) $\ell$ connecting the source $s$ to $u$ passes through $\pi_{u,\ell}$ as follows: $s \rightsquigarrow \pi_{u,\ell} \rightarrow u$. 
 The node $\pi_{u,\ell}$ coincides with $s$ if $\ell=2$ and $u$ is a successor of $s$; in this case, the path is simply $s \rightarrow u$.
 If there is no path of length $\ell$ connecting $s$ with $u$, then we write $\pi_{u,\ell} = \bot$.
 We show that the set 
 $\{\pi_{u,\ell}\ :\ 2\leq \ell \leq N,\ u\in Q\}$ 
 stores in just polynomial space all co-lexicographically smallest paths of any fixed length $2 \leq \ell \leq N$ from the source to any node $u$. We denote such path --- to be intended as a sequence $u_1 \rightarrow \dots \rightarrow u_\ell$ of states --- with $\alpha_\ell(u)$. 
 The node sequence $\alpha_\ell(u)$ can be obtained recursively (in $O(\ell)$ steps) as $\alpha_\ell(u) = \alpha_{\ell-1}(\pi_{u,\ell}) \rightarrow u$, where $\alpha_{1}(s) = s$ by convention.
 Note also that $\alpha_\ell(u)$ does not fully specify the sequence of edges (and thus labels) connecting those $\ell$ states, since two states may be connected by multiple (differently labeled) edges. However, the corresponding co-lexicographically smallest sequence $\lambda^-(\alpha_\ell(u))$ of $\ell-1$ labels is uniquely defined as follows:
 $$
 \left\{
 \begin{array}{ll}
    \lambda^-(\alpha_\ell(u)) = \mathrm{min}\{a\in\Sigma\ |\ \delta(s,a)=u\} & \mathrm{if}\ \ell=2 ,\\
    \lambda^-(\alpha_\ell(u)) = \lambda^-(\alpha_{\ell-1}(\pi_{u,\ell}) \rightarrow u) = \lambda^-(\alpha_{\ell-1}(\pi_{u,\ell})) \cdot \mathrm{min}\{a\in\Sigma\ |\ \delta(\pi_{u,\ell},a)=u\} & \mathrm{if}\ \ell > 2.
 \end{array}\right.
 $$
 

 It is not hard to see that each $\pi_{u,\ell}$ can be computed in $|Q|^{O(1)}$ time using dynamic programming. First, we set $\pi_{u,2} = s$ for all successors $u$ of $s$. Then, for $\ell = 3, \dots, N$:

$$
\pi_{u,\ell} = \underset{v\in \text{Pred}(u)}{\mathrm{argmin}} \Big(
  \lambda^-( \alpha_{\ell-1}(v) ) \cdot  \mathrm{min}\{a\in\Sigma\ |\ \delta(v,a)=u\} \Big)    
$$

where $\text{Pred}(u)$ is the set of all predecessors of $u$ and the $\mathrm{argmin}$ operator compares strings in co-lexicographic order. In the equation above, if none of the $\alpha_{\ell-1}(v)$ are well-defined (because there is no path of length $\ell-1$ from $s$ to $v$), then $\pi_{u,\ell} = \bot$. 
Since there are $|Q|\times N = |Q|^{O(k)}$ variables $\pi_{u,\ell}$ and each can be computed in time  $|Q|^{O(1)}$, overall Step (1) takes $|Q|^{O(k)}$ times.
This completes the description of Step (1).
 
\texttt{Step (2)}. Fix a $k$-tuple $u_1,\dots, u_k$ and a length $2\leq \ell \leq N$. Our goal is now to show how to compute the co-lexicographically largest string $\gamma$ of length $\ell -1$ labeling $k$ cycles of length (number of states) $\ell$  originating (respectively, ending) from (respectively, in) all the states $u_1,\dots, u_k$. Our final strategy will iterate over all such $k$-tuple of states (in time exponential in $k$) in order to find one satisfying the conditions of Corollary \ref{cor:conditions width}. 

Our goal can again be solved by dynamic programming. Let $u_1,\dots, u_k$ and $u'_1,\dots, u'_k$ be two $k$-tuples of states, and let $2\leq \ell \leq N$. Let moreover $\pi_{u_1,\dots, u_k, u'_1,\dots, u'_k, \ell}$ be the $k$-tuple $\langle u''_1,\dots, u''_k \rangle$ of states such that there exists a string $\gamma$ of length $\ell-1$ with the following properties:
\begin{itemize}
    \item For each $1\leq i \leq k$, there is a path $u_i \rightsquigarrow u''_i \rightarrow u'_i$ of length (number of nodes) $\ell$ labeled with $\gamma$, and
    \item $\gamma$ is the co-lexicographically largest string satisfying the above property. 
\end{itemize}

If such a string $\gamma$ does not exist, then we set $\pi_{u_1,\dots, u_k, u'_1,\dots, u'_k, \ell} = \bot$.

Remember that we fix $u_1,\dots, u_k$.
For $\ell = 2$ and each $k$-tuple $u'_1,\dots, u'_k$, it is easy to compute $\pi_{u_1,\dots, u_k, u'_1,\dots, u'_k, \ell}$:  this $k$-tuple is $\langle u_1,\dots, u_k \rangle$ (all paths have length 2) if and only if there exists $c\in\Sigma$ such that $u'_i \in \delta(u_i,c)$ for all $1\leq i\leq k$. Then, $\delta$ is formed by one character: the largest such $c$.

For $\ell >2$, the $k$-tuple $\pi_{u_1,\dots, u_k, u'_1,\dots, u'_k, \ell}$ can be computed as follows. Assume we have computed those variables for all lengths $\ell'<\ell$.
Note that for each such $\ell'<\ell$ and $k$-tuple $u''_1, \dots, u''_k$, the variables $\pi_{u_1,\dots, u_k, u''_1,\dots, u''_k, \ell'}$ identify $k$ paths $u_i \rightsquigarrow u''_i$ of length (number of nodes) $\ell'$. Let us denote with $\alpha_{\ell'}(u''_i)$ such paths, for $1\leq i\leq k$.

Then, $\pi_{u_1,\dots, u_k, u'_1,\dots, u'_k, \ell}$ is equal to  $\langle u''_1,\dots, u''_k \rangle$ maximizing co-lexicographically the string $\gamma'\cdot c$ defined as follows:

\begin{enumerate}
    \item $u'_i \in \delta(u''_i,c)$ for all $1\leq i\leq k$,  
    \item $\pi_{u_1,\dots, u_k, u''_1,\dots, u''_k, \ell-1} \neq \bot$, and
    \item $\gamma'$ is the co-lexicographically largest string labeling all the paths $\alpha_{\ell-1}(u''_i)$. Note that this string exists by condition (2), and it can be easily built by following those paths in parallel (choosing, at each step, the largest character labeling all the $k$ considered edges of the $k$ paths).
\end{enumerate}
 
If no $c\in \Sigma $ satisfies condition (1), or condition (2) cannot be met, then $\pi_{u_1,\dots, u_k, u'_1,\dots, u'_k, \ell} = \bot$.

Note that $\pi_{u_1,\dots, u_k, u_1,\dots, u_k, \ell}$ allows us to identify (if it exists) the largest string $\gamma$ of length $\ell-1$ labeling $k$ cycles originating and ending in each $u_i$, for $1\leq i \leq k$. 

Each tuple $\pi_{u_1,\dots, u_k, u'_1,\dots, u'_k, \ell}$ can be computed in $|Q|^{O(1)}$ time by dynamic programming (in order of increasing $\ell$), and there are $|Q|^{O(k)}$ such tuples to be computed (there are $|Q|^{O(k)}$ ways of choosing $u_1,\dots, u_k, u'_1,\dots, u'_k$, and $N \in O(|Q|^{2k}$). Overall, also Step (2) can therefore be solved in $|Q|^{O(k)}$ time. 

To sum up, we can check if the conditions of Corollary \ref{cor:conditions width} hold as follows:

\begin{enumerate}
    \item We compute $\pi_{u_i,\ell}$ for each $u\in Q$ and $\ell \leq N$. This identifies a string $\mu_u^\ell$ for each such pair $u\in Q$ and $\ell \leq N$: the co-lexicographically smallest one, of length $\ell$, labeling a path connecting $s$ with $u$.
    \item For each $k$-tuple $u_1, \dots, u_k$ and each $\ell \leq N$, we compute $\pi_{u_1,\dots, u_k, u_1,\dots, u_k, \ell}$. This identifies a string $\gamma_{u_1,\dots, u_k}^\ell$ for each such tuple $u_1,\dots, u_k$ and $\ell \leq N$: the co-lexicographically largest one, of length $\ell$, labeling $k$ cycles  originating and ending in each $u_i$, for $1\leq i \leq k$. 
    \item We identify the $k$-tuple $u_1, \dots, u_k$ and the lengths $\ell' < \ell \leq N$ (if they exist) such that $\mu_{u_i}^{\ell'} \prec \gamma_{u_1,\dots, u_k}^\ell$ for all $1\leq i \leq k$.
\end{enumerate}
 
The conditions of Corollary \ref{cor:conditions width} hold if and only if step 3 above succeeds for at least one $k$-tuple $u_1, \dots, u_k$ and lengths $\ell' < \ell \leq N$. Overall, including the main exponential search on $k$, the algorithm terminates in $|Q|^{O(width(\mathcal L(\mathcal A)))}$ time. \qed
 
\end{proof}

\section{Proofs of Section \ref{sec:minimal}}\label{app:nerode}

Let $\mt L \subseteq \Sigma^* $ be a language, and let $ \sim $ be an equivalence relation on $ \pf L $. We say that $ \sim $ \emph{respects $ \pf L $} if:
\begin{equation*}
    (\forall \alpha, \beta \in \pf L)(\forall \phi \in \Sigma^*)(\alpha \sim \beta \land \alpha \phi \in \pf L \to \beta \phi \in \pf L).
\end{equation*}
Now, let us define the right-invariant, $ \mathcal{P} $-consistent and $ \mathcal{P} $-convex refinements of an equivalence relation $ \sim $.
\begin{enumerate}
    \item Assume that $ \sim $ respects $ \pf  L $. For every $ \alpha, \beta \in \pf  L $, define:
\begin{equation*}
    \alpha \sim^r \beta \iff (\forall \phi \in \Sigma^*)(\alpha \phi \in \pf L \to \alpha \phi \sim \beta \phi).
\end{equation*}
    We say that $ \sim^r $ is the \emph{right-invariant refinement} of $ \sim $.
    \item  Let $ \mathcal{P} = \{U_1, \dots, U_p \} $ be a partition of $ \pf  L $. For every $ \alpha, \beta \in  \pf  L$, define:
\begin{equation*}
    \alpha \sim^{cs} \beta \iff (\alpha \sim \beta) \land (U_\alpha = U_\beta)
\end{equation*}
    We say that $ \sim^{cs} $ is the \emph{$ \mathcal{P} $-consistent refinement} of $ \sim $.
    \item Let $ \mathcal{P} = \{U_1, \dots, U_p \} $ be a partition of $ \pf L $. Assume that $ \sim $ is $ \mathcal{P} $-consistent. For every $ \alpha, \gamma   \in \pf L $, define:
    \begin{equation*}
\begin{split}
    & \alpha \sim^{cv} \gamma \iff (\alpha \sim \gamma) \land \\
    & \land (\forall \beta \in \pf L)(((U_{\alpha} = U_{\beta}) \land (\min \{\alpha, \gamma \} \prec \beta \prec \max \{\alpha, \gamma \}) \to \alpha \sim \beta).
\end{split}
\end{equation*}
    We say that $ \sim^{cv} $ is the \emph{$ \mathcal{P} $-convex refinement} of $ \sim $.
\end{enumerate}
It is easy to check that $ \sim^r $ is the coarsest right-invariant equivalence relation refining $ \sim $, $ \sim^{cs} $ is the coarsest $ \mathcal{P} $-consistent equivalence relation refining $ \sim $ and $ \sim^{cv} $ is the coarsest $ \mathcal{P} $-convex equivalence relation refining $ \sim $.

We wish to prove that any equivalence relation that respects $ Pref (\mathcal{L}) $ admits a coarsest refinement being $ \mathcal{P} $-consistent, $ \mathcal{P} $-convex and right-invariant at once, because then we will be able to define an equivalence relation inducing the minimum ($ \mathcal{P} $-sortable) DFA.
\begin{lemma}\label{lem:refinement}
Let $ \mathcal{L} \subseteq \Sigma^* $ be a language, and let $ \mathcal{P} $ be a partition of $ Pref (\mathcal{L}) $. If $\sim$ is a $ \mathcal{P} $-consistent and right-invariant equivalence relation on $ Pref (\mathcal{L}) $, then the relation $(\sim^{cv})^r$ is $ \mathcal{P} $-consistent, $ \mathcal{P} $-convex and right-invariant.
\end{lemma}
\begin{proof}
By definition $(\sim^{cv})^r$ is a right-invariant refinement. Moreover, $ \sim^{cv}$ and $(\sim^{cv})^r$ are $ \mathcal{P} $-consistent because they are refinements of the $ \mathcal{P} $-consistent equivalence relation $ \sim $. Let us prove that $ (\sim^{cv})^r $ is $\mathcal{P} $-convex. Assume that $ \alpha (\sim^{cv})^r \gamma $ and  $ \alpha \prec \beta \prec \gamma $ are such that $ U_\alpha  = U_\beta  $.  Being $(\sim^{cv})^r$ a $\mt P$-consistent relation, we have $ U_\alpha  = U_\beta =U_\gamma $. We must prove that $ \alpha (\sim^{cv})^r \beta $. Fix $ \phi \in \Sigma^* $ such that $ \alpha \phi \in \pf L $. We must prove that $ \alpha \phi \sim^{cv} \beta \phi $. Now, $ \alpha (\sim^{cv})^r \gamma $ implies $ \alpha \sim^{cv} \gamma $. Since $ \alpha \prec \beta \prec \gamma $ and $ U_\alpha = U_\beta = U_\gamma $, then the $ \mathcal{P} $-convexity of $ \sim^{cv} $ implies $ \alpha \sim^{cv} \beta $. In particular, $ \alpha \sim\beta $. Since $\sim$ is   right-invariant   we have  $ \alpha \phi \sim  \beta \phi $, and from the $\mt P$-consistency of $\sim$   we obtain    $ U_{\alpha \phi} = U_{\beta \phi} $. Moreover, $ \alpha (\sim ^{cv})^r \gamma $ implies $ \alpha \phi (\sim ^{cv})^r \gamma \phi $ by right-invariance, so $ \alpha \phi \sim ^{cv} \gamma \phi $. By $ \mathcal{P} $-convexity, from $ \alpha \phi \sim ^{cv} \gamma \phi $, $ U_{\alpha \phi} = U_{\beta \phi} $ and $ \alpha \phi \prec \beta \phi \prec \gamma \phi $ (since $ \alpha \prec \beta \prec \phi) $ we conclude $ \alpha \phi \sim ^{cv} \beta \phi $. \qed
\end{proof}

\begin{corollary}\label{cor:coarsest}
Let $ \mathcal{L} \subseteq \Sigma^* $ be a nonempty language, and let $ \mathcal{P} $ be a partition of $ Pref (\mathcal{L}) $. Let $ \sim $ be an equivalence relation that refines $ Pref (\mathcal{L}) $. Then, there exists a (unique) coarsest $ \mathcal{P} $-consistent, $ \mathcal{P} $-convex and right-invariant equivalence relation refining $ \sim $.
\end{corollary}
\begin{proof}
The equivalence relation $ (\sim^{cs})^r $ is $ \mathcal{P} $-consistent (because it is a refinement of the $ \mathcal{P} $-consistent equivalence relation $ \sim^{cs} $) and right-invariant (by definition it is a right-invariant refinement), so by Lemma \ref{lem:refinement} the  equivalence relation $(((\sim^{cs})^r)^{cv})^r$ is $ \mathcal{P} $-consistent, $ \mathcal{P} $-convex and right-invariant. Moreover, every $ \mathcal{P} $-consistent, $ \mathcal{P} $-convex and right-invariant equivalence relation refining $ \sim $ must also refine $(((\sim^{cs})^r)^{cv})^r$, so $(((\sim^{cs})^r)^{cv})^r$ is the coarsest $ \mathcal{P} $-consistent, $ \mathcal{P} $-convex and right-invariant equivalence relation refining $ \sim $. \qed
\end{proof}

Corollary \ref{cor:coarsest} allows us to give the following definition.

\begin{definition}
Let $ \mt L \subseteq \Sigma^* $ be a language, and let $ \mathcal{P} = \{U_1, \dots, U_p \} $ be a partition of $ \pf L $. Denote by $ \equiv_\mathcal{L}^\mathcal{P} $ the coarsest $ \mathcal{P} $-consistent, $ \mathcal{P} $-convex and right-invariant equivalence relation refining the Myhill-Nerode equivalence $ \equiv_\mathcal{L} $.
\end{definition}

Recall that, given a DFA $ \mathcal{A} = (Q, s, \delta, F) $, the equivalence relation $ \sim_\mathcal{A} $ on $ Pref (\mathcal{L(A)} $ is the one such that:
\begin{equation*}
    \alpha \sim_\mathcal{A} \beta  \iff \delta (s, \alpha) = \delta (s, \beta).
\end{equation*}

Here are the key properties of $ \sim_\mathcal{A} $.

\begin{lemma}\label{simA}
Let $ \mathcal{A} = (Q, s, \delta, F) $ be a $ \mathcal{P} $-sortable DFA, where $ \mathcal{P} = \{U_1, \dots, U_p \} $ is a partition of $ Pref (\mathcal{L(A)}) $. Then, $ \sim_\mathcal{A} $ refines $ \equiv_\mathcal{L} $, it respects $ \pf L $, it is $ \mathcal{P} $-consistent,$ \mathcal{P} $-convex, right-invariant, it has finite index and $ \mathcal{L(A)} $ is the union of some $ \sim_\mathcal{A} $-equivalence classes. In particular, $ \equiv_\mathcal{L}^\mathcal{P} $ is refined by $ \sim_\mathcal{A} $, it has finite index and $ \mathcal{L(A)} $ is the union of some $ \equiv_\mathcal{L}^\mathcal{P} $-equivalence classes.
\end{lemma}

\begin{proof}
Assume that $ \alpha \sim_\mathcal{A} \beta $. This means that if we read $ \alpha $ and $ \beta $ on $ \mathcal{A} $, then we reach the same state, which immediately implies that $ \mathcal{L(A)} $ is the union of some $ \sim_\mathcal{A} $-equivalence classes. Moreover, if $ \phi \in \Sigma^* $ satisfies $ \alpha \phi \in Pref (\mathcal{L(A)}) $, then it must be $ \beta \phi \in Pref (\mathcal{L(A)} $ and $ \delta(s, \alpha \phi) = \delta(s, \beta \phi) $, so proving that $ \sim_\mathcal{A} $ respects $ Pref (\mathcal{L(A)}) $, it is right-invariant and it refines $ \equiv_\mathcal{L} $. For every $ \alpha \in \pf L $ we have $ [\alpha]_\mathcal{A} = I_{\delta(s, \alpha)} $, which implies that $ \sim_\mathcal{A} $ is $ \mathcal{P} $-consistent and it has finite index. Moreover, $ \sim_\mathcal{A} $ is $ \mathcal{P} $-convex because for every $ \alpha \in Pref (\mathcal{L(A)}) $ we have that $ I_{\delta(s, \alpha)} $ is convex in $ U_\alpha $, because if $ q_1, \dots, q_k \in Q $ are such that $ U_\alpha = \cup_{i = 1}^k I_{q_i} $, then the $ I_{q_i} $'s are pairwise $ \le_\mathcal{A} $-comparable, being in the same $ \le_\mathcal{A} $-chain. Finally, $ \equiv_\mathcal{L}^\mathcal{P} $ is refined by $ \sim_\mathcal{A} $ because $ \equiv_\mathcal{L}^\mathcal{P} $ is the coarsest $ \mathcal{P} $-consistent, $ \mathcal{P} $-convex and right-invariant equivalence relation refining $ \equiv_\mathcal{L} $, so $ \equiv_\mathcal{L}^\mathcal{P} $ has finite index and $ \mathcal{L(A)} $ is the union of some $ \equiv_\mathcal{L}^\mathcal{P} $-equivalence classes because $ \sim_\mathcal{A} $ satisfies these properties. \qed
\end{proof}

We can now explain how to canonically build a $ \mathcal{P} $-sortable DFA starting from an equivalence relation.

\begin{lemma}\label{from_equiv_to_DFA}

Let $ L \subseteq \Sigma^* $ be a language, and let $ \mathcal{P} = \{U_1, \dots, U_p \} $ be a partition of $ \pf L $.  Assume that $ \mathcal{L} $ is the union of some classes of a $ \mathcal{P} $-consistent, $ \mathcal{P} $-convex, right invariant equivalence relation $ \sim $ on $ \pf L $ of finite index. Then, $ \mathcal{L} $ is recognized by a $ \mathcal{P} $-sortable DFA $ \mathcal{A_\sim} = (Q_\sim, s_\sim, \Sigma, E_\sim, F_\sim) $ such that:
\begin{enumerate}
\item $ |Q_\sim| $ is equal to the index of $ \sim $;
\item $ \sim_{A_\sim} $ and $ \sim $ are the same equivalence relation (in particular, $ |Q_\sim| $ is equal to the index of $ \sim_{A_\sim} $).
\end{enumerate}
Moreover, if $ \mathcal{B} $ is a $ \mathcal{P} $-sortable DFA that recognizes $ \mathcal{L} $, then $ \mathcal{A_{\sim_\mathcal{B}}} $ is isomorphic to $ \mathcal{B} $.
\end{lemma}

\begin{proof}
Define the DFA $ \mathcal{A_\sim} = (Q_\sim, E_\sim, \Sigma, s_\sim, F_\sim) $ as follows.
	\begin{itemize}
		\item $ Q_\sim = \{[\alpha]_\sim | \alpha \in \pf L\} $;
		\item $ s_\sim = [\epsilon]_\sim $, where $ \epsilon $ is the empty string;
		\item $ E_\sim = \{([\alpha]_\sim, [\alpha a]_\sim, a) | \alpha \in \Sigma^*, a \in \Sigma, \alpha a \in \pf L \}   $;
		\item $ F_\sim = \{[\alpha] | \alpha \in \mathcal{L} \} $.
	\end{itemize}
Since $ \sim $ is right-invariant, it has finite index and $ \mathcal{L} $ is the union of some $ \sim $-classes, then $ \mathcal{A_\sim} $ is a well-defined DFA (see also the proof of Theorem \ref{thm:equivalentholeproof}) and
\begin{equation}\label{eq13}
\alpha \in [\beta]_\sim \iff \delta_\sim (s_\sim, \alpha) = [\beta]_\sim.
\end{equation}
which implies that for every $ \alpha \in Pref(\mathcal{L}) $ it holds $ I_{[\alpha]_\sim} = [\alpha]_\sim $, and so $ \mathcal{L}(\mathcal{A_\sim}) = \mathcal{L} $.

For every $ i \in \{1, \dots, p \} $, define:
\begin{equation*}
    Q_i = \{[\alpha]_\sim | U_\alpha = U_i \}.
\end{equation*}

Notice that each $ Q_i $ is well-defined because $ \sim $ is $ \mathcal{P} $-consistent, and each $ Q_i $ is a $ \preceq_{\mathcal{A_\sim}} $-chain because $ \sim $ is $ \mathcal{P} $-convex, so $ \{Q_i\}_{i = 1}^p $ is a $ \preceq_{\mathcal{A_\sim}} $-chain partition of $ Q_\sim $.

From equation \ref{eq13} we obtain:
\begin{equation*}
\begin{split}
    Pref (\mathcal{L(A_\sim)})^i & = \{\alpha \in Pref (\mathcal{L(A_\sim)}) | \delta_\sim (s_\sim, \alpha) \in Q_i \} = \\
    & = \{\alpha \in Pref (\mathcal{L(A_\sim)}) | (\exists [\beta]_\sim \in Q_i | \alpha \in [\beta]_\sim) \} = \\
    & = \{\alpha \in Pref (\mathcal{L(A_\sim)}) | U_\alpha = U_i \} = U_i.
\end{split}
\end{equation*}
In other words, $ \mathcal{A}_\sim $ witnesses that $ \mathcal{L} $ is recognized by a $ \mathcal{P} $-sortable DFA. Moreover:
\begin{enumerate}
\item The number of states of $ \mathcal{A}_\sim $ is clearly equal to the index of $ \sim $.
\item By equation \ref{eq13}:
\begin{equation*}
    \alpha \sim_{A_\sim} \beta \iff \delta_\sim (s_\sim, \alpha) = \delta_\sim (s_\sim, \beta) \iff [\alpha]_\sim = [\beta]_\sim \iff \alpha \sim \beta
\end{equation*}
so $ \sim_{A_\sim} $ and $ \sim $ are the same equivalence relation.
\end{enumerate}
Finally, suppose $ \mathcal{B} $ is a $ \mathcal{P} $-sortable DFA that recognizes $ \mathcal{L} $. Notice that by lemma \ref{simA} we have that $ \sim_{\mathcal{B}} $ is a $ \mathcal{P} $-consistent, $ \mathcal{P} $-convex, right invariant equivalence relation on $ Pref(\mathcal{L}) $ of finite index such that $ \mathcal{L} $ is the union of some $ \sim_\mathcal{B} $-classes, so $ \mathcal{A_{\sim_\mathcal{B}}} $ is well-defined. Call $ Q_\mathcal{B} $ the set of states of $ \mathcal{B} $, and let $ \phi: Q_{\sim_\mathcal{B}} \to Q_\mathcal{B} $ be the function sending $ [\alpha]_{\sim_\mathcal{B}} $ into the state in $ Q_\mathcal{B} $ reached by reading $ \alpha $. Notice that $ \phi $ is well-defined because by the definition of $ \sim_{\mathcal{B}} $ we obtain that all strings in $ [\alpha]_{\sim_\mathcal{B}} $ reach the same state of $ \mathcal{B} $. It is easy to check that $ \phi $ determines an isomorphism between $ \mathcal{A_{\sim_\mathcal{B}}} $ and $ \mathcal{B} $. \qed
   \end{proof}

Here is the proof of our Myhill-Nerode theorem.

\paragraph*{\textbf{Statement of Theorem \ref{thm:MN}.}}    Let $ \mathcal{L} $ be a language. Let $ \mathcal{P} $ be a partition of $ \pf L $. The following are equivalent:
    \begin{enumerate}
        \item $ \mathcal{L} $ is recognized by a $ \mathcal{P} $-sortable DFA.
        \item $ \equiv_\mathcal{L}^\mathcal{P} $ has finite index.
        \item $ \mathcal{L} $ is the union of some classes of a $ \mathcal{P} $-consistent, $ \mathcal{P} $-convex, right invariant equivalence relation on $ \pf L $ of finite index.
    \end{enumerate}
Moreover, if one of the above statements is true (and so all the above statements are true), then there exists a unique minimum $ \mathcal{P} $-sortable DFA (that is, two $ \mathcal{P} $-sortable DFA having the minimum number of states must be isomorphic).

\paragraph*{\textbf{Proof}} 
$ (1) \to (2) $ It follows from lemma \ref{simA}

$ (2) \to (3) $ The desired equivalence relation is simply $ \equiv_\mathcal{L}^\mathcal{P} $. 

$ (3) \to (1) $ It follows from lemma \ref{from_equiv_to_DFA}.

Now, let us prove that the minimum automaton is $ \mathcal{A_{\equiv_\mathcal{L}^\mathcal{P}}} $ as defined in Lemma \ref{from_equiv_to_DFA}. First, $ \mathcal{A_{\equiv_\mathcal{L}^\mathcal{P}}} $ is well-defined because $ \equiv_\mathcal{L}^\mathcal{P} $ is $ \mathcal{P} $-consistent, $ \mathcal{P} $-convex an right-invariant by definition, and it has finite index and $ \mathcal{L(A)} $ is the union of some $ \equiv_\mathcal{L}^\mathcal{P} $-equivalence classes by lemma \ref{simA}. Now, the number of states of $ \mathcal{A_{\equiv_\mathcal{L}^\mathcal{P}}} $ is equal to the index of $ \equiv_\mathcal{L}^\mathcal{P} $, or equivalently, of $ \sim_{\mathcal{A_{\equiv_\mathcal{L}^\mathcal{P}}}} $. On the other hand, let $ \mathcal{B} $ be any $ \mathcal{P} $-sortable DFA recognizing $ \mathcal{L} $ non-isomorphic to $ \mathcal{A_{\equiv_\mathcal{L}^\mathcal{P}}} $. Then $ \sim_\mathcal{B} $ is a refinement of $ \equiv_\mathcal{L}^\mathcal{P} $ by Lemma \ref{from_equiv_to_DFA}, and it must be a strict refinement of $ \equiv_\mathcal{L}^\mathcal{P} $, otherwise $ \mathcal{A_{\equiv_\mathcal{L}^\mathcal{P}}} $ would be equal to $ A_{\sim_\mathcal{B}} $, which by Lemma \ref{from_equiv_to_DFA} is isomorphic to $ \mathcal{B} $, a contradiction. We conclude that the index of $ \equiv_\mathcal{L}^\mathcal{P} $ is smaller than the index of $ \sim_\mathcal{B} $, so again by Lemma \ref{from_equiv_to_DFA} the number of states of $ \mathcal{A_{\equiv_\mathcal{L}^\mathcal{P}}} $ is smaller than the number of states of $ A_{\sim_\mathcal{B}} $ and so of $ \mathcal{B} $. \qed


\end{document}